\documentclass[11pt,a4paper]{article}

\usepackage[T1]{fontenc}
\usepackage{lmodern}
\usepackage[margin=25mm]{geometry}
\usepackage{amsmath,amssymb,amsthm,mathtools}
\usepackage{microtype}
\usepackage{xcolor}
\usepackage{mdframed}
\usepackage{hyperref}
\usepackage{fancyhdr}

\hypersetup{
 colorlinks=true,
 allcolors=blue,
 pdftitle={The exact strong converse exponent for private communication over quantum channels},
 pdfauthor={Hao-Chung Cheng, Christoph Hirche, Marco Tomamichel}
}
\allowdisplaybreaks[2]
\numberwithin{equation}{section}

\newtheorem{theorem}{Theorem}[section]
\newtheorem{lemma}[theorem]{Lemma}
\newtheorem{proposition}[theorem]{Proposition}
\newtheorem{corollary}[theorem]{Corollary}
\theoremstyle{definition}
\newtheorem{definition}{Definition}[section]
\theoremstyle{remark}
\newtheorem*{remark}{Remark}

\definecolor{definitionbackground}{RGB}{237,244,252}
\definecolor{lemmabackground}{RGB}{239,247,239}
\definecolor{theorembackground}{RGB}{255,245,230}
\mdfdefinestyle{statement}{
 hidealllines=true,
 leftmargin=-8pt,rightmargin=-8pt,
 innerleftmargin=8pt,innerrightmargin=8pt,
 innertopmargin=4pt,innerbottommargin=4pt,
 skipabove=8pt,skipbelow=8pt
}
\surroundwithmdframed[style=statement,backgroundcolor=definitionbackground]{definition}
\surroundwithmdframed[style=statement,backgroundcolor=lemmabackground]{lemma}
\surroundwithmdframed[style=statement,backgroundcolor=lemmabackground]{proposition}
\surroundwithmdframed[style=statement,backgroundcolor=lemmabackground]{corollary}
\surroundwithmdframed[style=statement,backgroundcolor=theorembackground]{theorem}

\newcommand{\Tr}{\operatorname{Tr}}
\newcommand{\Sc}{\mathcal S}
\newcommand{\Sp}{\mathcal S_{>}}
\newcommand{\norm}[1]{\left\lVert#1\right\rVert}
\newcommand{\ket}[1]{\lvert#1\rangle}
\newcommand{\bra}[1]{\langle#1\rvert}

\newcommand{\arxiv}[1]{\href{https://arxiv.org/abs/#1}{arXiv:#1}}
\newcommand{\doi}[2]{\href{https://doi.org/#1}{#2}}

\title{The exact strong converse exponent\\for private communication over quantum channels}
\author{Hao-Chung Cheng\thanks{Department of Electrical Engineering and Graduate Institute of Communication Engineering, National Taiwan University, Taipei, Taiwan.}
\and Christoph Hirche\thanks{Institute for Information Processing  (tnt/L3S), Leibniz Universit\"at Hannover, Hannover, Germany.}
\and Marco Tomamichel\thanks{Department of Electrical and Computer Engineering and Centre for Quantum Technologies, National University of Singapore, Singapore.}}
\date{September 2026}

\begin{document}
\maketitle
\enlargethispage{3pt}

\begin{abstract}
We determine the exact strong converse exponent for secret-key
transmission and generation over general finite-dimensional quantum wiretap
channels, measuring reliability and secrecy jointly by squared fidelity to an
ideal secret key. We introduce a novel R\'enyi private information whose
regularization characterizes this exponent. An exact integral representation
in terms of ordinary private information gives a uniform continuity
bound, showing that the regularized quantity converges to private capacity
as the R\'enyi order approaches one. This establishes for any wiretap channel an exponential fidelity
decay at every rate above private capacity.
\end{abstract}

\section{Introduction}\label{sec:introduction}
Private communication over a quantum channel aims to transmit classical
information reliably to Bob while keeping it secret from Eve. When Eve
receives the channel environment, private capacity lies between quantum
and classical capacity: reliable quantum transmission also permits private
classical transmission, while removing the secrecy requirement gives
ordinary classical communication \cite{Devetak}. We study
\emph{secret-key transmission}: Alice is given a uniformly distributed
classical message, which Bob must recover without revealing it to Eve.
We assess reliability and secrecy jointly through the fidelity of the final state with
an ideal secret key. In the ideal state, Alice's message and Bob's estimate
agree and are independent of Eve.

We consider unassisted transmission over a fixed finite-dimensional
memoryless joint channel $\mathcal W:A\to BE$, where Bob receives $B$ and
Eve receives $E$. This includes the case in which $B$ and $E$ are
complementary outputs of a quantum channel.

The \emph{private capacity} $P(\mathcal W)$ is the supremum of rates
achievable by secret-key transmission with fidelity tending to one as the
blocklength increases. The coding theorems of Cai, Winter, and Yeung and
of Devetak \cite{Devetak,CaiWinterYeung} characterize it as
\begin{equation}
 P(\mathcal W)=\sup_{n\ge1}\frac1nP^{(1)}(\mathcal W^{\otimes n}),
 \qquad P^{(1)}(\mathcal W)=\sup_{\omega}
 \bigl\{I(X:B)_\omega-I(X:E)_\omega\bigr\},
 \label{eq:capacity}
\end{equation}
where the second supremum is over finite classical-quantum-quantum (cqq)
output states obtained by sending input ensembles through $\mathcal W$.
The mutual information difference measures Bob's advantage over Eve
in information about the classical label. The regularization allows
Alice to encode jointly across arbitrarily many channel uses.

The capacity theorem characterizes the rates at which asymptotically
perfect transmission is possible, but does not determine the optimal
fidelity above capacity. A \emph{strong converse} asserts that
the optimal fidelity tends to zero for every rate $R>P(\mathcal W)$.
An \emph{exponential strong converse} strengthens this conclusion to
exponential decay.

An $(n,M)$ secret-key transmission code
$\mathcal C$ transmits a uniform message from $M$ possibilities using
$n$ channel uses. Alice encodes the message in a state on $A^n$, and
Bob measures $B^n$ to produce an estimate. Let
$\rho_{M\widehat M E^n}$ be the state of the uniform message, Bob's
estimate, and Eve's output after decoding. The code's squared fidelity is
\begin{equation}
 F(\mathcal C;\mathcal W):=\max_{\sigma_{E^n}}
 F(\rho_{M\widehat M E^n},\kappa_{M\widehat M}\otimes\sigma_{E^n}),
 \qquad \kappa_{M\widehat M}=\frac1M\sum_{m=1}^M|mm\rangle\langle mm|,
 \label{eq:key-fidelity}
\end{equation}
where the maximum is over states $\sigma_{E^n}$ on Eve's output and
$F(\rho,\sigma)=\|\sqrt\rho\sqrt\sigma\|_1^2$. For complementary
channel outputs, this is equivalent to the squared-fidelity criterion
used by Wilde, Tomamichel, and Berta \cite[Eq.~(3.3)]{WTB}, with the
optimization over Eve's reference state written explicitly. We use
the same criterion for general joint channels $\mathcal W:A\to BE$.
The best squared fidelity at blocklength $n$ and rate at least $R$ is
\begin{equation}
 F^\star(n,R;\mathcal W):=\sup_{\mathcal C}F(\mathcal C;\mathcal W),
 \label{eq:optimal-fidelity}
\end{equation}
where the supremum is over all $(n,M)$ secret-key transmission codes with
$\log M\ge nR$.

The \emph{strong converse exponent} is the optimal
decay rate, defined by
\begin{equation}
 E_{\mathrm{sc}}(\mathcal W,R)
 :=\liminf_{n\to\infty}-\frac1n\log F^\star(n,R;\mathcal W).
 \label{eq:operational-exponent}
\end{equation}
Determining the exponent exactly requires a converse bound valid for
every code and a sequence of codes attaining the same exponent.

Previous results for private communication include pretty strong
converses for degradable quantum channels and degraded c-q wiretap
channels \cite{MorganWinter,WinterPretty}, and converse bounds for
anti-degradable channels, whose private capacity is zero, under a
specified reliability--secrecy error region \cite{KhanianHirche}.
General meta-converse and geometric R\'enyi methods give further bounds
on private communication \cite{WTB,FangFawzi}, but the rates above
which they guarantee fidelity decay can exceed private capacity.
Recently, Wilde~\cite{WildeDegradable} established an exponential strong
converse at private capacity for finite-dimensional degradable and
anti-degradable quantum channels under the same joint-fidelity criterion.
For classical wiretap channels, strong converses are known for degraded
channels with public feedback \cite{HayashiTyagiWatanabe} and for
general channels under total-variation secrecy
\cite{GravesWong,TyagiWatanabe}. Secrecy exponents below capacity
have also been studied \cite{BastaniTelatarMerhav}; these concern the
approach to secrecy below capacity rather than the decay of performance
above capacity.

Related questions have been studied for other communication tasks.
For entanglement-assisted classical communication, Gupta and
Wilde~\cite{GuptaWilde} prove an exponential strong converse, and
Li and Yao~\cite{LiYaoEA} determine the exact strong converse exponent.
Mosonyi and Ogawa~\cite{MosonyiOgawa} characterize the exact
strong converse exponent for c-q channel coding. Recently, Cheng and Tomamichel~\cite{CT} proved strong converses at the classical and quantum capacities
of general quantum channels using integral representations of R\'enyi
information along tilted states. These representations provide a
method for controlling R\'enyi information uniformly in the blocklength,
which we adapt here to private communication.

We introduce a R\'enyi private information that converges to ordinary
private information at order one and admits an operational expression
in terms of the secret key fraction at order infinity. Its integral
representation in terms of ordinary private information along tilted
states yields continuity at order one after regularization. Together
with a one-shot fidelity converse, this establishes an exponential
strong converse above private capacity, following the approach of
Cheng and Tomamichel \cite{CT}.
For achievability, we follow the approach of Mosonyi and Ogawa
\cite{MosonyiOgawa}: a change of channel first gives an achievable
exponent in terms of a quantum variational formula. Local pinching
and a double-blocking argument then relate this exponent to our
regularized R\'enyi private information, giving a matching bound.

The paper is organized as follows.
Section~\ref{sec:main-results} presents the main results and their
operational consequences.
Section~\ref{sec:radius} establishes the properties and representations
of the R\'enyi private information, and
Section~\ref{sec:capacity-proof} proves continuity at order one after
regularization.
Sections~\ref{sec:converse} and~\ref{sec:coding} prove the converse and
matching achievability bounds, including the extension to secret-key
generation.
Section~\ref{sec:classical} proves the classical variational formula.
Appendix~\ref{sec:criteria} compares the different error criteria.

\section{Main results and discussion}\label{sec:main-results}
This section introduces the R\'enyi private information and its one-shot
converse bound, states the exact strong converse exponent and continuity
under regularization, and discusses special cases and other error criteria.
All logarithms are natural.

\subsection{R\'enyi private information}
The joint fidelity criterion retains correlations between Bob's decoding
outcome and Eve's state, so its exponent can depend on the joint channel
and not only on its marginals. We introduce a R\'enyi private information
that accounts for these correlations.

Let $\omega_{XBE}=\sum_{x\in\mathcal X}p_x|x\rangle\langle x|
\otimes\omega_{BE,x}$ be a cqq state with finite alphabet $\mathcal X$,
omitting zero-probability letters. Restrict $B$ and $E$
to the supports of the average marginals $\omega_B$ and $\omega_E$.

We write $\Sc(R)$ for the states on $R$ and $\Sp(R)$ for the strictly positive ones.

\begin{definition}[R\'enyi private information]\label{def:radius}
For $\alpha>1$ and a cqq state $\omega_{XBE}$ as above, define
\begin{align}
 P_\alpha^{(1)}(\omega)
 &:=\frac{\alpha}{\alpha-1}\log
 \sup_{\sigma_E}\inf_{\tau_B}\sup_{\boldsymbol\eta_B}
 \inf_{\boldsymbol\nu_E}
 G_\alpha(\tau_B,\sigma_E,\boldsymbol\nu_E,\boldsymbol\eta_B),
   \label{eq:radius-trace} \\
G_\alpha(\tau_B,\sigma_E,\boldsymbol\nu_E,\boldsymbol\eta_B)
 &:=\sum_xp_x\Tr\!\left[\omega_{BE,x}\left(
 \tau_B^{\frac{1-\alpha}{2\alpha}}
 \eta_{B,x}^{\frac{\alpha-1}{\alpha}}
 \tau_B^{\frac{1-\alpha}{2\alpha}}
 \otimes
 \sigma_E^{\frac{\alpha-1}{2\alpha}}
 \nu_{E,x}^{\frac{1-\alpha}{\alpha}}
 \sigma_E^{\frac{\alpha-1}{2\alpha}}
 \right)\right].
 \label{eq:trace-objective}
\end{align}
Here $\tau_B\in\Sp(B)$, $\sigma_E\in\Sc(E)$,
$\boldsymbol\nu_E=(\nu_{E,x})_x\in\Sp(E)^{\mathcal X}$, and
$\boldsymbol\eta_B=(\eta_{B,x})_x\in\Sc(B)^{\mathcal X}$.
\end{definition}
The R\'enyi private information is additive:
for finite cqq states $\omega_{XBE}$ and $\zeta_{X'B'E'}$,
\begin{equation}
 P_\alpha^{(1)}(\omega\otimes\zeta)
 =P_\alpha^{(1)}(\omega)+P_\alpha^{(1)}(\zeta) \,.
 \label{eq:results-ensemble-additivity}
\end{equation}
Lemma~\ref{lem:ensemble-additivity} establishes this property.
The channel optimization also allows correlated input ensembles across
uses, so this identity does not eliminate the need for regularization.
The quantity also satisfies data processing
in opposite directions at the two terminals: applying a channel to Bob's
system cannot increase it, while applying a channel to Eve's system
cannot decrease it (Lemma~\ref{lem:local-data-processing}).
Equivalent norm and dual quasi-norm representations are given in
Lemma~\ref{lem:dual-norm}.

The limit as $\alpha\searrow1$ is expressed through the usual
von Neumann entropy, quantum relative entropy, and mutual information:
\begin{align}
 H(A)_\rho&=-\Tr\rho_A\log\rho_A,\\
D(\rho\|\sigma)&=\Tr\rho(\log\rho-\log\sigma),\\
I(A:B)_\rho&=D(\rho_{AB}\|\rho_A\otimes\rho_B).
\end{align}
Here $D(\rho\|\sigma)=+\infty$ unless
$\operatorname{supp}\rho\subseteq\operatorname{supp}\sigma$.

\begingroup\mdfsetup{nobreak=true}
\begin{lemma}\label{lem:ensemble-order-one}
For every finite cqq state $\omega_{XBE}$ in Definition~\ref{def:radius},
\begin{equation}
 \lim_{\alpha\searrow1}P_\alpha^{(1)}(\omega)
 =I(X:B)_\omega-I(X:E)_\omega.
\label{eq:ensemble-order-one}
\end{equation}
\end{lemma}
\endgroup
A proof is given in Section~\ref{sec:ensemble-order-one}.
The first-order expansion also explains the four optimizations.
Expanding $G_\alpha$ from \eqref{eq:trace-objective} as
$\alpha\searrow1$, with faithful reference and auxiliary states held
fixed, gives
\begin{align}
\frac{\alpha}{\alpha-1}\log G_\alpha
 ={}&I(X:B)_\omega-I(X:E)_\omega
   +D(\omega_B\|\tau_B)-D(\omega_E\|\sigma_E)
   \label{eq:intro-first-order-start}\\
 &+\sum_xp_x\bigl[D(\omega_{E,x}\|\nu_{E,x})
                  -D(\omega_{B,x}\|\eta_{B,x})\bigr]+O(\alpha-1).
\label{eq:intro-first-order}
\end{align}
Roughly speaking, the infima select $\tau_B=\omega_B$ and
$\nu_{E,x}=\omega_{E,x}$, while the suprema select
$\sigma_E=\omega_E$ and $\eta_{B,x}=\omega_{B,x}$.
All relative-entropy corrections then vanish, leaving the mutual
information difference. At other orders, compatible optimizing states
are the corresponding marginals of the tilted ensemble introduced
in Section~\ref{sec:integral}.

At infinite order, the quantity has a direct fidelity interpretation.
Let $d=|\mathcal X|$, and let
$\kappa_{X\widehat X}=d^{-1}\sum_x|xx\rangle\langle xx|$ be an ideal
uniform key. We call
\begin{equation}
 F_{\mathrm{key}}(\omega)
 :=\sup_{\substack{\mathcal D:B\to\widehat X\\\sigma_E\in\Sc(E)}}
 F\!\left((\mathrm{id}_X\otimes\mathcal D\otimes\mathrm{id}_E)(\omega_{XBE}),
          \kappa_{X\widehat X}\otimes\sigma_E\right)
 \label{eq:secret-key-fraction}
\end{equation}
the maximally achievable secret key fraction, where $\mathcal D$ ranges over measurement
channels with outcome alphabet $\mathcal X$. It measures the best joint
reliability and secrecy attainable by decoding Bob's system. This is
the secret-key analogue of the maximal achievable singlet fraction
\cite{KoenigRennerSchaffner}: both optimize a
local decoding operation on Bob's system and compare its output with
an ideal shared resource. Here the target is a uniform classical key
independent of Eve, in place of a maximally entangled state.
\begin{lemma}
\label{lem:ensemble-order-infinity}
For every finite cqq state $\omega_{XBE}$ in Definition~\ref{def:radius},
\begin{equation}
 P_\infty^{(1)}(\omega)
 :=\lim_{\alpha\to\infty}P_\alpha^{(1)}(\omega)
 =\log d+\log F_{\mathrm{key}}(\omega).
 \label{eq:infinity-fidelity}
\end{equation}
\end{lemma}
The input distribution $(p_x)$ is arbitrary; the comparison key is
uniform. Section~\ref{sec:ensemble-order-infinity} proves the identity
and gives a primal--dual pair of semidefinite programs for the same
quantity.

Under the faithfulness assumptions below, the optimizers in
Definition~\ref{def:radius} also give a tilted cqq state by applying local
operators and normalizing the original state. Its conditional
states on Bob's and Eve's systems are the optimizing auxiliary states,
and their averages are the optimizing reference states. These tilted
states yield an integral representation in terms of ordinary private
information.
\begin{lemma}[Integral representation]\label{lem:integral}
Let $\omega_{XBE}$ be a finite cqq state whose conditional marginals
$\omega_{B,x}$ and $\omega_{E,x}$ are faithful. For each $u \in(0,1)$,
let $\Theta(u)$ be the tilted cqq state defined in
\eqref{eq:tilted-ensemble}, constructed from the optimizers in
Lemma~\ref{au-optimizer-theorem} at parameter $u$.
Then, for every $\alpha>1$,
\begin{equation}
 \boxed{P_\alpha^{(1)}(\omega)=\frac{\alpha}{\alpha-1}
 \int_0^{\frac{\alpha-1}{\alpha}}
 \bigl[I(X:B)_{\Theta(u)}-I(X:E)_{\Theta(u)}\bigr]\,du.}
 \label{eq:intro-integral}
\end{equation}
\end{lemma}
Thus $P_\alpha^{(1)}$ is an average of private information
along the tilted family, analogous to the integral representations
for public and quantum communication in \cite{CT}.
The tilted states need not be channel outputs; controlling their
departure from the channel image is the main step in the continuity
proof in Section~\ref{sec:capacity-proof}.

When $\omega_{XBE}$ is classical, with joint law $p_{XBE}$, the same
quantity has the variational form
\begin{equation}
 P_\alpha^{(1)}(p_{XBE})
 =\sup_{q \ll p}
 \left\{I(X:B)_q-I(X:E)_q
       -\frac{\alpha}{\alpha-1}D(q_{XBE}\|p_{XBE})\right\}.
 \label{eq:intro-classical-private-information}
\end{equation}
Thus it is the largest private information under a changed joint law $q_{XBE}$,
with relative entropy penalizing its departure from the original law.
Proposition~\ref{prop:classical-radius} proves this specialization.

\subsection{Strong converse exponent}\label{sec:channel-exponent}
A quantum channel $\mathcal W:A\to BE$ is a
completely positive, trace-preserving linear map from operators on
Alice's input system $A$ to operators on the joint output $BE$, with
Bob receiving $B$ and Eve receiving $E$. For a finite input ensemble
$\{p_x,\rho_{A,x}\}_{x\in\mathcal X}$, with $\rho_{A,x}\in\Sc(A)$,
we denote the induced cqq output state by
\begin{equation}
 \omega_{XBE}
 :=\sum_xp_x|x\rangle\langle x|\otimes\mathcal W(\rho_{A,x}).
 \label{eq:channel-output-ensemble}
\end{equation}

\begin{definition}[R\'enyi private information of a channel]\label{def:channel-radius}
For $\alpha>1$ and a finite-dimensional channel $\mathcal W:A\to BE$, define
\begin{equation}
 P_\alpha^{(1)}(\mathcal W)
 :=\sup_{\{p_x,\rho_{A,x}\}}P_\alpha^{(1)}(\omega_{XBE}),
 \label{eq:renyi-one-letter-channel}
\end{equation}
where the supremum is over finite input ensembles and $\omega_{XBE}$
is their channel output as in \eqref{eq:channel-output-ensemble}.
\end{definition}
For general quantum channels, private information can be strictly
superadditive already at order one, with gains that persist over
arbitrarily many channel uses \cite{ElkoussStrelchuk}.

For a classical channel $W(b,e|a)$, the variational formula specializes to
\begin{align}
P_\alpha^{(1)}(W)=\sup_{q_{XABE}}\Bigl\{
 &I(X:B)_q-I(X:E)_q -\frac{\alpha}{\alpha-1}
 D(q_{BE|XA}\|W_{BE|A}\mid q_{XA})\Bigr\}.
\label{eq:classical-penalty-capacity}
\end{align}
Here $X$ ranges over finite alphabets and allows stochastic encoding
into the channel input $A$. The conditional relative entropy is
$D(q_{XABE}\|q_{XA}W)$, so the penalty measures departure from the
channel law. Section~\ref{sec:classical} proves this representation.

For integers $n,M\ge1$, an $(n,M)$ secret-key transmission code
$\mathcal C$ consists of encoding states $\rho_{A^n,m}\in\Sc(A^n)$,
$m=1,\ldots,M$, and a decoding POVM $\{\Lambda_m\}_{m=1}^M$ on
$B^n$. Its conditional output states are
\begin{equation}
 \omega_{B^nE^n,m}:=\mathcal W^{\otimes n}(\rho_{A^n,m}).
 \label{eq:code-conditional-output}
\end{equation}
Bob's measurement produces the decoded state
$\rho_{M\widehat M E^n}$ for a uniformly distributed message.
The R\'enyi private information bounds the fidelity of every such
code. Here $F(\mathcal C;\mathcal W)$ is the squared fidelity
achieved by the particular code $\mathcal C$ with an ideal secret key,
optimized over Eve's reference state.
\begin{theorem}[One-shot converse]\label{thm:intro-one-shot}
Let $\mathcal C$ be a $(1,M)$ secret-key transmission code for
$\mathcal W$.
For every $\alpha>1$,
\begin{equation}
 F(\mathcal C;\mathcal W)
 \le \exp\!\left\{-\frac{\alpha-1}{\alpha}
       [\log M-P_\alpha^{(1)}(\mathcal W)]\right\}.
 \label{eq:intro-one-shot}
\end{equation}
\end{theorem}
This is Proposition~\ref{prop:oneshot}.

\begin{definition}\label{def:regularized-private-information}
For $\alpha>1$ and a finite-dimensional channel $\mathcal W:A\to BE$, define
\begin{equation}
 P_\alpha(\mathcal W):=\sup_{k\ge1}\frac1k
       P_\alpha^{(1)}(\mathcal W^{\otimes k}).
 \label{eq:regularized-private-information}
\end{equation}
\end{definition}
The quantity $F^\star(n,R;\mathcal W)$ is the best squared fidelity
achievable by any code of blocklength $n$ and rate at least $R$.
Since $P_\alpha^{(1)}(\mathcal W^{\otimes n})\le nP_\alpha(\mathcal W)$,
applying Theorem~\ref{thm:intro-one-shot} to $\mathcal W^{\otimes n}$
gives, for every~$n$,
\begin{equation}
 F^\star(n,R;\mathcal W)
 \le \exp\!\left\{-n\frac{\alpha-1}{\alpha}
             [R-P_\alpha(\mathcal W)]\right\}.
 \label{eq:intro-regularized-converse}
\end{equation}
Continuity at order one also holds after regularization.
\begin{theorem}[Asymptotic Continuity]\label{thm:main}
For every finite-dimensional channel $\mathcal W:A\to BE$,
\begin{equation}
 \lim_{\alpha\searrow1}P_\alpha(\mathcal W)=P(\mathcal W).
 \label{eq:renyi-capacity-continuity}
\end{equation}
\end{theorem}

For any rate $R>P(\mathcal W)$, this continuity lets us choose a
fixed $\alpha>1$ such that $P_\alpha(\mathcal W)<R$.
Equation~\eqref{eq:intro-regularized-converse} then bounds the optimal
fidelity by $e^{-nc}$, where
$c=(\alpha-1)[R-P_\alpha(\mathcal W)]/\alpha>0$.
Thus the one-shot converse and continuity already establish the
exponential strong converse above private capacity.

To determine the exact decay rate, we optimize the converse bound over
$\alpha>1$ and prove a matching achievability bound in
Section~\ref{sec:coding}. This gives the following characterization.

\begin{samepage}
\begin{theorem}[Exact strong converse exponent]\label{thm:exact-quantum-exponent}
For every finite-dimensional channel $\mathcal W:A\to BE$ and $R\ge0$,
the liminf in \eqref{eq:operational-exponent} is a limit, and
\begin{equation}
 E_{\mathrm{sc}}(\mathcal W,R)=\sup_{\alpha>1}
 \frac{\alpha-1}{\alpha}[R-P_\alpha(\mathcal W)].
 \label{eq:exact-quantum-exponent}
\end{equation}
\end{theorem}
\end{samepage}

The theorem follows by combining the regularized converse
\eqref{eq:intro-regularized-converse}, optimized over $\alpha$, with
the matching achievability bound in Theorem~\ref{shell-interpolation}.
Lemma~\ref{lem:operational-convexity} establishes existence of the
operational limit. Achievability starts from a direct quantum
generalization of the classical variational formula
\eqref{eq:intro-classical-private-information}.

\begin{definition}[Variational private information]
\label{alt:variational-definition}
For $\alpha>1$ and a finite cqq state $\omega_{XBE}$, define
\begin{equation}
 \widetilde P_\alpha^{(1)}(\omega)
 :=\sup_{\theta_{XBE}\ll\omega_{XBE}}
 \left\{I(X:B)_\theta-I(X:E)_\theta
                   -\frac{\alpha}{\alpha-1}D(\theta\|\omega)\right\}.
 \label{alt:variational-value}
\end{equation}
Here $\theta\ll\omega$ means inclusion of supports and $\theta$ is optimised over cqq states.
\end{definition}
For classical states this equals $P_\alpha^{(1)}$, whereas the two
quantities need not agree for quantum states. The achievability proof
in Section~\ref{sec:coding} has four steps.

First, ordinary private coding on the modified channel specified by
$\theta$, followed by a relative-entropy bound on the fidelity cost of
returning to the original channel, achieves the exponent expressed
through $\widetilde P_\alpha^{(1)}$
(Proposition~\ref{alt:change-channel}).
Second, universal reference states replace the outer reference-state
optimizations at a cost that vanishes per channel use. Local pinching
then compares $\widetilde P_\alpha^{(1)}$ on pinched tensor-power
ensembles with $P_\alpha^{(1)}$ on the original ensemble. Third, invariant coding tests transfer the construction
back to the original channel, again at a vanishing cost per use.
These steps are established in Propositions~\ref{alt:controlled-comparison}
and~\ref{alt:operational-bridge}.
Finally, apply the construction to ensembles over arbitrary input
blocks and optimize over their size. This blocking step yields the
regularized quantity $P_\alpha(\mathcal W)$ and the achievability
bound in Theorem~\ref{shell-interpolation}.

\subsection{Secret-key generation and other error criteria}
In unassisted secret-key generation, Alice prepares the key register
jointly with the channel input. The target is still the uniform ideal
key in \eqref{eq:key-fidelity}. We show that a generation code can be
converted to a secret-key transmission code of the same blocklength
and key size with at most a factor-four loss in squared fidelity.
Consequently, both tasks have the exponent in
Theorem~\ref{thm:exact-quantum-exponent}.

For comparison with separate reliability and secrecy requirements
\cite{Devetak,WTB}, we use the trace distance
$T(\rho,\sigma)=\tfrac12\|\rho-\sigma\|_1$.
For the same $(n,M)$ code, define
\begin{equation}
 p_{\rm succ}=\frac1M\sum_m\Tr\Lambda_m\omega_{B^n,m},
 \qquad
 \delta_T=\min_{\sigma_{E^n}\in\Sc(E^n)}\frac1M\sum_m
                 T(\omega_{E^n,m},\sigma_{E^n}).
 \label{eq:gap-definition}
\end{equation}
The quantity $[p_{\rm succ}-\delta_T]_+$, with
$[a]_+=\max\{a,0\}$, measures how far the sum of decoding error
$1-p_{\rm succ}$ and secrecy error $\delta_T$ lies below one.
The same method gives exponential decay of this quantity above
private capacity; Corollary~\ref{cor:quantitative-strong-converse} gives
the corresponding strong converse.

The exact decay rate under separate reliability and secrecy requirements
need not coincide with that for squared joint-key fidelity;
Appendix~\ref{sec:criteria} gives examples distinguishing them.

\section{R\'enyi private information and its properties}\label{sec:radius}
We put $d_R=\dim R$ and write $\mathcal L(R)$
for the linear operators on $R$ and $\mathcal L^\dagger$ for the
Hilbert--Schmidt adjoint of a linear map $\mathcal L$.
The binary entropy is
\begin{equation}
 h_2(t)=-t\log t-(1-t)\log(1-t),\qquad 0\le t\le1,
 \label{eq:binary-entropy}
\end{equation}
with $0\log0=0$.
Negative powers denote support inverses.

\subsection{Norm and dual quasi-norm representations}
For each letter, define the positive conditional map and its adjoint by
\begin{align}
 \mathcal T_x(Z)&:=\Tr_E[(I_B\otimes Z)\omega_{BE,x}],
 \label{eq:conditional-map}\\
 \mathcal T_x^\dagger(Y)&:=\Tr_B[(Y\otimes I_E)\omega_{BE,x}].
 \label{eq:conditional-adjoint}
\end{align}
For effects $Z$ and $Y$, these give Bob's unnormalized state conditioned
on Eve's effect and Eve's unnormalized state conditioned on Bob's effect,
respectively. The maps retain the correlations of the joint state.

For the equivalent norm representations, it is convenient to write
\begin{equation}
 u=\frac{\alpha-1}{\alpha}\in(0,1),\qquad
 \alpha=\frac1{1-u},\qquad
 \beta=\frac1{1+u}=\frac{\alpha}{2\alpha-1}.
 \label{eq:alpha-coordinate}
\end{equation}
As $\alpha\searrow1$, we have $\beta\nearrow1$.
\begin{lemma}\label{lem:dual-norm}
For $\alpha>1$, with $u$ and $\beta$ as in \eqref{eq:alpha-coordinate},
we have
\begin{align}
P_\alpha^{(1)}(\omega)
 &=\frac1u\log\inf_{\tau_B\in\Sp(B)}
 \sup_{\sigma_E\in\Sc(E)}\inf_{\{\nu_{E,x}\in\Sp(E)\}}
 \sum_xp_x\left\|\tau_B^{-u/2}\mathcal T_x
 (\sigma_E^{u/2}\nu_{E,x}^{-u}\sigma_E^{u/2})
 \tau_B^{-u/2}\right\|_\alpha
 \label{eq:radius-norm}\\
 &=\frac1u\log\inf_{\tau_B\in\Sp(B)}
 \sup_{\substack{\sigma_E\in\Sc(E)\\\{\eta_{B,x}\in\Sc(B)\}}}
 \sum_xp_x\left\|\sigma_E^{u/2}
 \mathcal T_x^\dagger(\tau_B^{-u/2}\eta_{B,x}^u\tau_B^{-u/2})
 \sigma_E^{u/2}\right\|_\beta.
\label{eq:radius-primal}
\end{align}
Here $\|Z\|_q=(\Tr|Z|^q)^{1/q}$ denotes the Schatten norm for
$q\ge1$ and quasi-norm for $0<q<1$.
For fixed $\tau_B,\sigma_E$, the infimum over $\boldsymbol\nu_E$ and
the supremum over $\boldsymbol\eta_B$ can be interchanged.
After these conditional auxiliary states have been optimized, the common
objective before taking the logarithm is convex in $\tau_B$ and concave
in $\sigma_E$. The outer infimum and supremum can therefore be interchanged.
\end{lemma}
\begin{proof}
We first identify the inner value at fixed $\tau_B\in\Sp(B)$ and
$\sigma_E\in\Sc(E)$. Put
\begin{equation}
 K_x(\nu)=\tau_B^{-u/2}\mathcal T_x
 (\sigma_E^{u/2}\nu^{-u}\sigma_E^{u/2})\tau_B^{-u/2}.
\end{equation}
The unweighted letter term in $G_\alpha$ is $\Tr\eta^uK_x(\nu)$.
Schatten duality maximizes it to $\|K_x(\nu)\|_\alpha$, with
$\eta=K_x(\nu)^\alpha/\Tr K_x(\nu)^\alpha$ when $K_x(\nu)\ne0$.
For $A\ge0$ and $0<u<1$, the variational identity
\begin{equation}
 \|A\|_{1/(1+u)}
 =\inf_{\nu\in\Sp(E)}\Tr A\nu^{-u}
\end{equation}
holds, with the infimum approached by faithful approximations to
$A^{1/(1+u)}/\Tr A^{1/(1+u)}$ when $A\ne0$.
The function $(\eta,\nu)\mapsto\Tr\eta^uK_x(\nu)$ is continuous,
concave in $\eta\in\Sc(B)$, and convex in $\nu\in\Sp(E)$, by
operator concavity of $t^u$ and operator convexity of $t^{-u}$.
Since $\Sc(B)$ is compact, Sion's minimax theorem \cite{Sion} gives
\begin{align}
 \sup_\eta\inf_\nu\Tr\eta^uK_x(\nu)
 &=\inf_\nu\sup_\eta\Tr\eta^uK_x(\nu)\\
 &=\inf_\nu\|K_x(\nu)\|_\alpha\\
 &=\sup_{\eta\in\Sc(B)}
 \left\|\sigma_E^{u/2}\mathcal T_x^\dagger
 (\tau_B^{-u/2}\eta^u\tau_B^{-u/2})\sigma_E^{u/2}\right\|_\beta.
 \label{eq:fixed-reference-norm-duality}
\end{align}
The conditional optimizations separate in the weighted sum. Denote their
common value by $F(\tau_B,\sigma_E)$; Definition~\ref{def:radius} is
$u^{-1}\log\sup_{\sigma_E}\inf_{\tau_B}F(\tau_B,\sigma_E)$.

To exchange these outer optimizations, first fix $\sigma_E$ and a letter
$x$, and set $\Psi(Z)=\mathcal T_x(\sigma_E^{u/2}Z\sigma_E^{u/2})$.
Hiai's negative-power convexity theorem
\cite[Theorem~5.3(1)]{HiaiII}, with powers $p=q=-u$ and scalar function
$t\mapsto t^\alpha$, shows that
\begin{equation}
 f(\tau,\nu):=
 \Tr\left[\tau^{-u/2}\Psi(\nu^{-u})\tau^{-u/2}\right]^\alpha
 \label{alt:hiai-trace-function}
\end{equation}
is jointly convex, since $\alpha(1-u)=1$.
A strictly positive perturbation of $\Psi$, followed by a limit, handles
maps with singular outputs. The function $f$ is homogeneous of degree
$-2u\alpha$. A positive convex function homogeneous of degree $-d<0$
has concave $f^{-1/d}$: its superlevel sets are convex, and degree-one
homogeneity turns quasiconcavity into concavity. Composing with the convex
decreasing function $t\mapsto t^{-2u}$ proves joint convexity of
$f^{1/\alpha}$. Summing and minimizing jointly over the conditional
states $\nu_{E,x}$ therefore proves convexity of $F$ in $\tau_B$.
Limits include zero letter values.

For concavity in $\sigma_E$, fix $\tau_B$ and define
$\mathcal M_x(A)=\mathcal T_x^\dagger(\tau_B^{-u/2}A\tau_B^{-u/2})$.
The dual expression gives
\begin{equation}
 h_x(K):=\sup_{\eta\in\Sc(B)}
 \left\|K^{1/2}\mathcal M_x(\eta^u)K^{1/2}\right\|_\beta.
 \label{eq:auxiliary-concavity}
\end{equation}
Hiai's positive-power theorem \cite[Theorem~1.1(1)]{Hiai}, with powers
$(u,1)$ and trace exponent $\beta=1/(1+u)$, makes the trace of the
$\beta$ power jointly concave in $(\eta,K)$. Its supremum over $\eta$
is concave and homogeneous of degree $\beta$ in $K$.
Taking its $1/\beta$ power gives a degree-one homogeneous,
quasiconcave function, hence the concave function $h_x$.
It is also increasing in the positive-semidefinite order.
Operator concavity of $t^u$ thus makes
$F(\tau_B,\sigma_E)=\sum_xp_xh_x(\sigma_E^u)$ concave in $\sigma_E$.
The dual expression makes it continuous on the compact state space
$\Sc(E)$. Sion's theorem now exchanges the outer infimum and supremum,
giving both displayed representations. This exchange applies to $F$
after the conditional optimizations; no concavity of the unoptimized
trace expression in $\sigma_E$ is required.
\end{proof}

\subsection{Local data processing}\label{sec:local-data-processing}
\begingroup\mdfsetup{nobreak=true}
\begin{lemma}[Local data processing]\label{lem:local-data-processing}
For every finite cqq state $\omega_{XBE}$, $\alpha>1$, and channels
$\mathcal N:B\to B'$ and $\mathcal M:E\to E'$, we have
\begin{align}
 P_\alpha^{(1)}\bigl((\mathrm{id}_X\otimes\mathcal N\otimes
 \mathrm{id}_E)(\omega)\bigr)
 &\le P_\alpha^{(1)}(\omega),
 \label{eq:data-processing-bob}\\
 P_\alpha^{(1)}\bigl((\mathrm{id}_X\otimes\mathrm{id}_B\otimes
 \mathcal M)(\omega)\bigr)
 &\ge P_\alpha^{(1)}(\omega).
 \label{eq:data-processing-eve}
\end{align}
\end{lemma}
\endgroup
\begin{proof}
Let $u$ and $\beta$ be as in \eqref{eq:alpha-coordinate}.
We use data processing for the sandwiched R\'enyi divergence at orders
$\alpha>1$ and $\beta\in(1/2,1)$ \cite[Theorem~1]{FrankLieb}.
As throughout, restrict all systems to their average marginal supports;
then $\mathcal N(\tau_B)$ is faithful on the output support whenever
$\tau_B\in\Sp(B)$.

For Bob, the conditional map after the channel is
$\mathcal N\circ\mathcal T_x$. Fix the reference states in
the norm representation \eqref{eq:radius-norm} and put
$A_x=\mathcal T_x(\sigma_E^{u/2}\nu_{E,x}^{-u}\sigma_E^{u/2})$.
Data processing at order $\alpha$ gives
\begin{equation}
 \bigl\|\mathcal N(\tau_B)^{-u/2}\mathcal N(A_x)
 \mathcal N(\tau_B)^{-u/2}\bigr\|_\alpha
 \le\bigl\|\tau_B^{-u/2}A_x\tau_B^{-u/2}\bigr\|_\alpha.
 \label{eq:local-dpi-alpha-norm}
\end{equation}
Sum with weights $p_x$, minimize over $\boldsymbol\nu_E$, and maximize
over $\sigma_E$. The output infimum over Bob's reference includes every
$\mathcal N(\tau_B)$, so taking the input infimum over $\tau_B$ and
then $u^{-1}\log$ proves \eqref{eq:data-processing-bob}.

For Eve, use Lemma~\ref{lem:dual-norm}; the adjoint conditional map
after the channel is $\mathcal M\circ\mathcal T_x^\dagger$.
Fix $\tau_B,\sigma_E,\{\eta_{B,x}\}$ and put
$A_x=\mathcal T_x^\dagger
(\tau_B^{-u/2}\eta_{B,x}^u\tau_B^{-u/2})$.
Since $(1-\beta)/\beta=u$ and $\beta<1$, data processing at order
$\beta$ gives the reversed quasi-norm inequality
\begin{equation}
 \bigl\|\mathcal M(\sigma_E)^{u/2}\mathcal M(A_x)
 \mathcal M(\sigma_E)^{u/2}\bigr\|_\beta
 \ge\bigl\|\sigma_E^{u/2}A_x\sigma_E^{u/2}\bigr\|_\beta.
 \label{eq:local-dpi-beta-norm}
\end{equation}
Both norm inequalities apply to arbitrary positive $A_x$: normalize
when nonzero and use trace preservation, with the zero case immediate.
The output maximization over Eve's reference includes every
$\mathcal M(\sigma_E)$. Summing and maximizing over
$\sigma_E,\{\eta_{B,x}\}$ therefore gives a value no smaller after Eve's
channel for each fixed $\tau_B$. Taking the infimum over $\tau_B$ and
then $u^{-1}\log$ proves \eqref{eq:data-processing-eve}.
\end{proof}

\subsection{The order-one limit for a fixed ensemble}\label{sec:ensemble-order-one}
\begin{proof}[Proof of Lemma~\ref{lem:ensemble-order-one}]
For faithful $\tau_B,\sigma_E,\nu_{E,x}$, write
\begin{equation}
 K_x(u)=\tau_B^{-u/2}\mathcal T_x
 (\sigma_E^{u/2}\nu_{E,x}^{-u}\sigma_E^{u/2})\tau_B^{-u/2},
 \qquad \alpha(u)=\frac1{1-u}.
\end{equation}
At $u=0$, $K_x(0)=\omega_{B,x}$ and $\Tr K_x(0)=1$.
Differentiating the matrix powers and the Schatten order gives
\begin{align}
 \left.\frac{d}{du}\|K_x(u)\|_{\alpha(u)}\right|_{u=0}
 &=\Tr\omega_{B,x}(\log\omega_{B,x}-\log\tau_B)
   +\Tr\omega_{E,x}(\log\sigma_E-\log\nu_{E,x})\\
 &=D(\omega_{B,x}\|\tau_B)-D(\omega_{E,x}\|\sigma_E)
   +D(\omega_{E,x}\|\nu_{E,x}).
\end{align}
The contribution $\Tr\omega_{B,x}\log\omega_{B,x}$ comes from varying
the norm order, since $\alpha'(0)=1$. Summing with weights $p_x$,
taking the logarithm, and using
\begin{equation}
 \sum_xp_xD(\omega_{R,x}\|\zeta_R)
 =I(X:R)_\omega+D(\omega_R\|\zeta_R),\qquad R=B,E,
\end{equation}
gives the first-order expansion after maximizing over Bob's conditional
auxiliary states in \eqref{eq:intro-first-order-start}--\eqref{eq:intro-first-order}.
To justify the optimized limit, work with \eqref{eq:radius-norm} and
restrict its three reference families $\tau_B,\sigma_E,\nu_{E,x}$ to
\begin{equation}
 \Sc_\varepsilon(R)=\{\rho\in\Sc(R):\rho\ge\varepsilon I_R/d_R\},
 \qquad 0<\varepsilon<1,
\end{equation}
and denote the resulting value by $P_{\alpha,\varepsilon}^{(1)}$.
On these compact sets, $K_x(u)$ is analytic of constant rank
$\operatorname{rank}\omega_{B,x}$, with its positive eigenvalues uniformly
separated from zero near $u=0$. The expansion therefore has a uniform
remainder $O(u)=O(\alpha-1)$, including for singular signal marginals.

Put $c=1-\varepsilon$ and $\rho_\varepsilon=c\rho+\varepsilon I_R/d_R$.
The inequalities $\rho_\varepsilon^u\ge c^u\rho^u$ and
$\rho_\varepsilon^{-u}\le c^{-u}\rho^{-u}$ (for faithful $\rho$ in the
second), together with $\|C^\dagger C\|_q=\|CC^\dagger\|_q$, control
all three restrictions. First restricting $\sigma_E$ in
\eqref{eq:radius-primal} changes the optimized sum by a factor in
$[c^u,1]$; then restricting $\tau_B$ and the family $\boldsymbol\nu_E$
in \eqref{eq:radius-norm} costs at most $c^{-u}$ each. Hence, uniformly
in $\alpha>1$ and $\omega$,
\begin{equation}
 \log c\le P_{\alpha,\varepsilon}^{(1)}(\omega)-P_\alpha^{(1)}(\omega)
 \le-2\log c.
 \label{eq:reference-smoothing}
\end{equation}
Each relative-entropy minimum in the optimized leading term lies in
$[0,-\log c]$, since $D(\rho\|\rho_\varepsilon)\le-\log c$.
Thus this leading term differs from
$I(X:B)_\omega-I(X:E)_\omega$ by at most $-2\log c$.
Taking $\alpha\searrow1$ and then $\varepsilon\searrow0$ in
\eqref{eq:reference-smoothing} proves the claim.
\end{proof}

\begin{corollary}[Support reduction and state continuity]\label{cor:state-continuity}
For fixed $\alpha>1$ and finite systems $X,B,E$, $P_\alpha^{(1)}$ is
continuous on cqq states and unchanged by restricting $B$ and $E$
to the supports of their average marginals.
\end{corollary}
\begin{proof}
Operator Jensen for $t^{-u}$ and $t^u$ shows that compression to the
average supports, followed by normalization, improves the respective
minimizations and maximization in \eqref{eq:radius-norm} and
\eqref{eq:radius-primal}. Compress $\tau_B$ and $\sigma_E$ first,
then $\nu_{E,x}$; a zero compression of $\sigma_E$ gives zero objective
and cannot maximize. Faithful block-diagonal extensions approximate
the restricted references, proving support invariance.

For fixed $\varepsilon>0$, compact optimization makes
$P_{\alpha,\varepsilon}^{(1)}$ continuous in the blocks
$p_x\omega_{BE,x}$, including zero blocks. The uniform bound
\eqref{eq:reference-smoothing} makes $P_\alpha^{(1)}$ their uniform
limit, hence continuous.
\end{proof}

\subsection{The order-infinity limit}\label{sec:ensemble-order-infinity}
\begin{proof}[Proof of Lemma~\ref{lem:ensemble-order-infinity}]
Put $Q_u=e^{uP_{1/(1-u)}^{(1)}(\omega)}$ for $0<u<1$, and define $Q_1$
by the right side of \eqref{eq:radius-primal} before taking its logarithm,
with $u=1$ and $\beta=1/2$. The inner minimax argument of Lemma~\ref{lem:dual-norm}
extends to this endpoint, replacing the $\alpha$-norm by the operator
norm. Restricting only $\tau_B$ to $\Sc_\varepsilon(B)$ gives a value
$Q_{u,\varepsilon}$ satisfying
\begin{equation}
 Q_u\le Q_{u,\varepsilon}\le(1-\varepsilon)^{-u}Q_u,
 \qquad 0<u\le1,
 \label{eq:infinity-reference-smoothing}
\end{equation}
by the same smoothing argument as in \eqref{eq:reference-smoothing}.
On these compact reference sets the dual objective is jointly continuous
through $u=1$, so $Q_{u,\varepsilon}\to Q_{1,\varepsilon}$.
Letting $\varepsilon\searrow0$ proves $Q_u\to Q_1$ and hence
$P_\infty^{(1)}=\log Q_1$, for arbitrary signal states.

For positive operators we use the same convention
$F(A,C)=\|\sqrt A\sqrt C\|_1^2$, without a normalization requirement.
Since $\|\sigma^{1/2}A\sigma^{1/2}\|_{1/2}=F(A,\sigma)$, this endpoint is
\begin{align}
 Q_1&=\inf_{\tau_B\in\Sp(B)}\max_{\sigma_E\in\Sc(E)}
             \sum_xp_x a_x(\tau_B,\sigma_E),
 \label{eq:infinity-reference-form}\\
 a_x(\tau_B,\sigma_E)
 &:=\max_{\eta_B\in\Sc(B)}
 F\!\left(\mathcal T_x^\dagger(\tau_B^{-1/2}\eta_B\tau_B^{-1/2}),
          \sigma_E\right).
 \label{eq:infinity-letter-value}
\end{align}
The operator-norm form also gives
\begin{equation}
 a_x(\tau_B,\sigma_E)
 =\inf_{\substack{A_E>0\\\mathcal T_x(A_E)\le\tau_B}}
                  \Tr\sigma_E A_E^{-1}.
 \label{eq:infinity-letter-primal}
\end{equation}
For faithful $\sigma_E$, this follows by setting
$A_E=\sigma_E^{1/2}\nu_E^{-1}\sigma_E^{1/2}/\lambda$, where $\lambda$
is the letter's operator norm. Conversely, a feasible $A_E$ gives
$\nu_E=\sigma_E^{1/2}A_E^{-1}\sigma_E^{1/2}/\Tr\sigma_E A_E^{-1}$
and an operator norm at most $\Tr\sigma_E A_E^{-1}$.
Mixing $\sigma_E$ with $I_E/d_E$ extends the identity to the boundary.
The right side is affine in $\sigma_E$ before minimization and convex
in $A_E$. Sion's theorem therefore interchanges the maximum over
$\sigma_E$ with the infimum over the family $(A_{E,x})$, giving
\begin{equation}
 Q_1=\inf_{\substack{\tau_B\in\Sp(B),\;A_{E,x}>0\\
                     \mathcal T_x(A_{E,x})\le\tau_B\ \forall x}}
                  \left\|\sum_xp_xA_{E,x}^{-1}\right\|_\infty.
\end{equation}
If the displayed norm is $r$, rescale $A_{E,x}$ by $r$ and set
$H_B=r\tau_B$. Conversely, divide a feasible $H_B$ and all $A_{E,x}$
by $\Tr H_B$; $H_B$ is faithful on Bob's average support.
Thus
\begin{equation}
 Q_1=\inf_{\substack{H_B\ge0,\;A_{E,x}>0\\
       \mathcal T_x(A_{E,x})\le H_B\ \forall x\\
       \sum_xp_xA_{E,x}^{-1}\le I_E}}\Tr H_B.
 \label{eq:infinity-common-upper-bound}
\end{equation}
This also gives a reduced primal form in terms of a noncommutative
maximum. For a finite family of positive operators $(B_x)_x$, write
\cite{ChengAchievability,Watrous}
\begin{equation}
 \bigvee_x B_x
 :=\underset{H\ge B_x\ \forall x}{\operatorname{argmin}}\Tr H.
 \label{eq:noncommutative-maximum}
\end{equation}
Set $\Gamma_{E,x}=p_xA_{E,x}^{-1}$. These strictly positive effects satisfy
$\sum_x\Gamma_{E,x}\le I_E$. Distributing the residual among the effects
makes them a POVM and decreases every inverse. Minimizing over $H_B$
in \eqref{eq:infinity-common-upper-bound} therefore gives
\begin{equation}
 P_\infty^{(1)}(\omega)
 =\log\inf_{\substack{\{\Gamma_{E,x}\}\,\mathrm{POVM}\\\Gamma_{E,x}>0}}
 \Tr\!\left[\bigvee_x p_x\mathcal T_x(\Gamma_{E,x}^{-1})\right].
 \label{eq:infinity-primal}
\end{equation}
Here Eve's POVM is only a variational reference. In the classical case,
its effects are conditional probabilities $q(x|e)$, so the inverses
are reciprocal reference weights.

Taking Schur complements gives the primal SDP
\begin{equation}
 Q_1=\inf_{H_B,\{A_{E,x},C_{E,x}\}}\Tr H_B,
 \label{eq:infinity-sdp-objective}
\end{equation}
subject to
\begin{align}
 H_B&\ge\mathcal T_x(A_{E,x})\quad\text{for every }x,
 \label{eq:infinity-sdp-domination}\\
 \sum_x C_{E,x}&\le I_E,
 \label{eq:infinity-sdp-normalization}\\
 \begin{pmatrix}
 A_{E,x}&\sqrt{p_x}I_E\\
 \sqrt{p_x}I_E&C_{E,x}
 \end{pmatrix}&\ge0\quad\text{for every }x.
 \label{eq:infinity-sdp-blocks}
\end{align}
The primal variables are Hermitian. Since zero-probability
letters are omitted, the block constraints already enforce
$A_{E,x}>0$. Its dual is
\begin{equation}
 Q_1=\max_{\{\Gamma_{B,x},Z_{E,x}\},S_E}
 \left\{2\sum_x\sqrt{p_x}\,\operatorname{Re}\Tr Z_{E,x}-\Tr S_E\right\},
 \label{eq:infinity-sdp-dual-objective}
\end{equation}
subject to
\begin{align}
 \Gamma_{B,x}&\ge0\quad\text{for every }x,
 \label{eq:infinity-sdp-dual-positivity}\\
 \sum_x\Gamma_{B,x}&=I_B,
 \label{eq:infinity-sdp-dual-normalization}\\
 \begin{pmatrix}
 \mathcal T_x^\dagger(\Gamma_{B,x})&Z_{E,x}\\
 Z_{E,x}^\dagger&S_E
 \end{pmatrix}&\ge0\quad\text{for every }x.
 \label{eq:infinity-sdp-dual-blocks}
\end{align}
Here $S_E$ is Hermitian and the $Z_{E,x}$ are arbitrary operators.
The diagonal blocks of the dual multipliers for
\eqref{eq:infinity-sdp-blocks} are fixed to
$\mathcal T_x^\dagger(\Gamma_{B,x})$ and $S_E$ by stationarity in
$A_{E,x}$ and $C_{E,x}$; stationarity in $H_B$ gives
$\sum_x\Gamma_{B,x}=I_B$. A sign change in the off-diagonal blocks gives
the displayed objective.
The primal is strictly feasible: take $A_{E,x}=2I_E$,
$C_{E,x}=3p_xI_E/4$, and $H_B=3I_B$.
Slater's theorem therefore gives equality of the optimal values and
attainment of the dual maximum, including for singular signal states.

For fixed $\Gamma_{B,x}$ and $S_E$, maximizing $\operatorname{Re}\Tr Z_{E,x}$
under \eqref{eq:infinity-sdp-dual-blocks} gives
$\sqrt{F(\mathcal T_x^\dagger(\Gamma_{B,x}),S_E)}$ \cite{Watrous}.
Writing $S_E=t\sigma_E$ with $t\ge0$ and $\sigma_E\in\Sc(E)$ then yields
\begin{align}
 Q_1
 &=\max_{\substack{\{\Gamma_{B,x}\}\,\mathrm{POVM}\\
                   \sigma_E\in\Sc(E),\;t\ge0}}
 \left\{2\sqrt t\sum_x\sqrt{p_x}\,
 \sqrt{F(\mathcal T_x^\dagger(\Gamma_{B,x}),\sigma_E)}-t\right\}
 \label{eq:infinity-sdp-dual-reduction}\\
 &=\left(\max_{\substack{\{\Gamma_{B,x}\}\,\mathrm{POVM}\\
                         \sigma_E\in\Sc(E)}}
 \sum_x\sqrt{p_x}\,\sqrt{F(\mathcal T_x^\dagger(\Gamma_{B,x}),\sigma_E)}
 \right)^2.
 \label{eq:infinity-sdp-fidelity-value}
\end{align}
The POVM $\{\Gamma_{B,x}\}$ defines a measurement channel
$\mathcal D:B\to\widehat X$. The block-diagonal fidelity identity gives
\begin{equation}
 \sqrt{F\!\left((\mathrm{id}_X\otimes\mathcal D\otimes\mathrm{id}_E)(\omega),
                \kappa_{X\widehat X}\otimes\sigma_E\right)}
 =\frac1{\sqrt d}\sum_x\sqrt{p_x}\,
       \sqrt{F(\mathcal T_x^\dagger(\Gamma_{B,x}),\sigma_E)}.
 \label{eq:infinity-key-blocks}
\end{equation}
Only the blocks with equal classical labels contribute, since the
ideal key vanishes on all other blocks. Thus
$Q_1=dF_{\mathrm{key}}(\omega)$, and taking logarithms proves
\eqref{eq:infinity-fidelity}.
\end{proof}

\subsection{Fixed-point identities}\label{sec:optimizers}

The optimizers in Definition~\ref{def:radius} satisfy conditional
stationarity and average-state identities. These identify them as the
conditional and average marginals of the tilted state constructed below.

\begin{lemma}[Fixed-point identities]\label{au-optimizer-theorem}
Let $\omega_{XBE}$ be a finite cqq state as in Definition~\ref{def:radius}.
Fix $\alpha>1$ and write $u=(\alpha-1)/\alpha$.
Assume that every $\omega_{B,x}$ and $\omega_{E,x}$ is
faithful on $B$ and $E$, respectively. The variational expression
\begin{equation}
\begin{split}
 &\sup_{\sigma_E}\inf_{\tau_B}
   \sup_{\boldsymbol\eta_B}\inf_{\boldsymbol\nu_E}
   \sum_x p_x\Tr\!\left[\omega_{BE,x}\left(
   \tau_B^{-u/2}\eta_{B,x}^{u}\tau_B^{-u/2}
   \otimes
   \sigma_E^{u/2}\nu_{E,x}^{-u}\sigma_E^{u/2}
   \right)\right]
\end{split}
\label{eq:optimizer-u-objective}
\end{equation}
has a faithful optimizing tuple
$\tau_B^\star,\sigma_E^\star,\boldsymbol\nu_E^\star,\boldsymbol\eta_B^\star$.
The optimization domains are those in Definition~\ref{def:radius}.
Let $\ell_x$ be its $x$th trace term, without the factor $p_x$,
evaluated at this tuple, and put $\beta=1/(1+u)$. Then
\begin{align}
\ell_x&=
 \left\|(\sigma_E^\star)^{u/2}
 \mathcal T_x^\dagger\!\left((\tau_B^\star)^{-u/2}(\eta_{B,x}^\star)^u
                    (\tau_B^\star)^{-u/2}\right)
 (\sigma_E^\star)^{u/2}\right\|_\beta,\label{eq:optimal-letter-norm}\\
 r_x&=\frac{p_x\ell_x}{\sum_y p_y\ell_y}.
\label{eq:optimal-letter-weights}
\end{align}
The distribution $(r_x)$ is normalized, and
\begin{equation}
 \log\sum_xp_x\ell_x
 =\frac{\alpha-1}{\alpha}P_\alpha^{(1)}(\omega).
 \label{eq:optimized-value-u}
\end{equation} The minimizing conditional states on Eve's
system satisfy
\begin{equation}
 \nu_{E,x}^\star=
 \left[
 \frac{(\sigma_E^\star)^{u/2}
 \mathcal T_x^\dagger\!\left((\tau_B^\star)^{-u/2}(\eta_{B,x}^\star)^u
                    (\tau_B^\star)^{-u/2}\right)
 (\sigma_E^\star)^{u/2}}{\ell_x}
 \right]^\beta.
 \label{eq:optimal-eve-state}
\end{equation}
For every $x$, the optimizing auxiliary state on Bob's system satisfies
\begin{equation}
 (\tau_B^\star)^{-u/2}\mathcal T_x
 \bigl((\sigma_E^\star)^{u/2}(\nu_{E,x}^\star)^{-u}(\sigma_E^\star)^{u/2}\bigr)
 (\tau_B^\star)^{-u/2}
 =\ell_x(\eta_{B,x}^\star)^{1-u}.
 \label{au-power-stationarity}
\end{equation}
The optimizing reference states satisfy
\begin{equation}
 \tau_B^\star=\sum_xr_x\eta_{B,x}^\star,\qquad
        \sigma_E^\star=\sum_xr_x\nu_{E,x}^\star.
 \label{au-selfconsistency}
\end{equation}
For every fixed $\tau_B\in\Sp(B)$, the maximizing Eve reference and
auxiliary states in \eqref{eq:radius-primal} are unique and faithful.
The objective after this maximization is continuously differentiable in
$(u,\tau_B)$. All optimizing states and the quantities $\ell_x,r_x,\nu_{E,x}^\star$
may depend on $u$ and $\omega$.
\end{lemma}

\begin{proof}
We first establish the unique faithful maximizing states at fixed Bob
reference, then derive the average-state identities by stationarity.
Fix an arbitrary faithful Bob reference
$\tau_B$ in \eqref{eq:radius-primal}, and define
$\mathcal M_x(A)=\mathcal T_x^\dagger(\tau_B^{-u/2}A\tau_B^{-u/2})$.
Its value at $I_B$ and its adjoint's value at $I_E$ are faithful.
Use the concave, degree-one homogeneous functions $h_x$ from
\eqref{eq:auxiliary-concavity}, whose properties were established in the
proof of Lemma~\ref{lem:dual-norm}.

The objective after maximizing the $\eta_{B,x}$ is
$\sum_xp_xh_x(\sigma_E^u)$.
If two distinct states $\sigma_0,\sigma_1$ maximized it, put
\begin{equation}
 \widehat\sigma=
 \left[\frac{\sigma_0^u+\sigma_1^u}{2}\right]^{1/u}.
\end{equation}
Concavity of $h_x$ preserves the optimal value before normalization.
Strict convexity of $\Tr A^{1/u}$ gives $\Tr\widehat\sigma<1$.
Normalizing $\widehat\sigma$ multiplies its objective by
$(\Tr\widehat\sigma)^{-u}>1$, a contradiction.
Thus $\sigma_E$ is unique.

It is faithful as well. If $Q\ne0$ is its kernel projection,
concavity and homogeneity imply superadditivity, so replacing
$\sigma_E^u$ by $\sigma_E^u+tQ$ increases the objective by at least
$t\sum_xp_xh_x(Q)>0$.
Positivity follows by taking $\eta=I_B/d_B$ in \eqref{eq:auxiliary-concavity}.
But
\begin{equation}
 \Tr(\sigma_E^u+tQ)^{1/u}=1+t^{1/u}\Tr Q=1+o(t),
\end{equation}
so normalization cannot offset this linear improvement.

For fixed faithful $\sigma_E$, the function
$A\mapsto\|\sigma_E^{u/2}\mathcal M_x(A)\sigma_E^{u/2}\|_{1/(1+u)}$
is concave and homogeneous. This follows directly from the variational
formula
$\|C\|_{1/(1+u)}=\inf_{\nu>0,\,\Tr\nu=1}\Tr C\nu^{-u}$, whose optimizer is
$\nu=C^{1/(1+u)}/\Tr C^{1/(1+u)}$.
The same $u$-power averaging proves uniqueness of each $\eta_{B,x}$.
The same kernel perturbation proves its faithfulness: a nonzero
kernel projection has strictly positive value because
$\mathcal M_x^\dagger(I_E)>0$ and $\sigma_E>0$.

Compactness and uniqueness make the maximizing family continuous in
$(u,\tau_B)$. In a neighborhood of any fixed $(u,\tau_B)$ it consequently
stays in a compact interior set. The objective and its first derivatives
are continuous there. Comparing optimal values using each other's
maximizers therefore permits differentiation with the maximizer held
fixed and proves continuous differentiability of the maximized value.
To prove attainment of the outer infimum, choose $\eta_{B,x}=I_B/d_B$ and
$\sigma_E=I_E/d_E$ in \eqref{eq:radius-primal}. Since
$\|C\|_{1/(1+u)}\ge\Tr C$, the maximized sum there is at least
\begin{equation}
 (d_Bd_E)^{-u}\Tr\omega_B\tau_B^{-u},\qquad
 \omega_B=\sum_xp_x\omega_{B,x}>0.
\end{equation}
This diverges as $\tau_B$ approaches the boundary of the state space.
The value at $\tau_B=I_B/d_B$ is finite, which confines minimizers to
a compact interior set. Both statements hold uniformly when $u$ varies
over a compact subinterval of $(0,1)$: the smallest eigenvalue of
$\omega_B$ is positive and the negative exponent is bounded away from
zero. This is the uniform compactness needed in the integral proof.

Choose an optimizing tuple
$\tau_B^\star,\sigma_E^\star,\{\eta_{B,x}^\star\}$ in the dual expression.
The outer minimax identity in Lemma~\ref{lem:dual-norm} makes these
reference states a saddle pair for Definition~\ref{def:radius}:
the maximin is attained by upper semicontinuity on $\Sc(E)$, and its
optimizer must be the unique maximizer at $\tau_B^\star$.
The unique minimizing states $\nu_{E,x}^\star$ are the faithful normalized
powers in \eqref{eq:optimal-eve-state}. Their substitution in the
trace objective gives \eqref{eq:optimal-letter-norm}.
To prove \eqref{au-power-stationarity} and \eqref{au-selfconsistency},
use these states and the weights in \eqref{eq:optimal-letter-weights}.
At fixed $u$, the differential of the norm in
\eqref{eq:optimal-letter-norm} is the trace pairing with $(\nu_{E,x}^\star)^{-u}$.
Let $C_x$ denote the left side of \eqref{au-power-stationarity}.
The stationarity equation for $\eta_{B,x}^\star$ is
$D(\eta^u)_{\eta_{B,x}^\star}[C_x]=\lambda_x I$.
In an eigenbasis of $\eta_{B,x}^\star$, the off-diagonal entries of this
derivative are the entries of $C_x$ multiplied by the strictly positive
divided differences of $t^u$. Hence $C_x$ is diagonal in that basis,
and the diagonal equations give
$C_x=(\lambda_x/u)(\eta_{B,x}^\star)^{1-u}$.
Since $\Tr C_x(\eta_{B,x}^\star)^u=\ell_x$, normalization gives
$\lambda_x=u\ell_x$ and proves \eqref{au-power-stationarity}.

Using \eqref{au-power-stationarity}, the gradient of the objective with
respect to $(\tau_B^\star)^{-u/2}$ is
\begin{equation}
 \left(\sum_xp_x\ell_x\right)
 \left[\left(\sum_xr_x\eta_{B,x}^\star\right)(\tau_B^\star)^{u/2}
         +(\tau_B^\star)^{u/2}\left(\sum_xr_x\eta_{B,x}^\star\right)\right].
\end{equation}
The constraint is $\Tr((\tau_B^\star)^{-u/2})^{-2/u}=1$.
Taking the trace of its multiplier equation against $(\tau_B^\star)^{-u/2}$ determines
the multiplier and yields
\begin{equation}
 (\tau_B^\star)^{u/2}\left(\sum_xr_x\eta_{B,x}^\star-\tau_B^\star\right)
 +\left(\sum_xr_x\eta_{B,x}^\star-\tau_B^\star\right)(\tau_B^\star)^{u/2}=0.
\end{equation}
Multiplication by the parenthesized Hermitian difference and taking
the trace forces that difference to vanish.
For Eve, apply the same calculation to $(\sigma_E^\star)^{u/2}$ under the constraint
$\Tr((\sigma_E^\star)^{u/2})^{2/u}=1$. The definition
\eqref{eq:optimal-eve-state} gives
\begin{equation}
 (\sigma_E^\star)^{-u/2}\left(\sum_xr_x\nu_{E,x}^\star-\sigma_E^\star\right)
 +\left(\sum_xr_x\nu_{E,x}^\star-\sigma_E^\star\right)(\sigma_E^\star)^{-u/2}=0,
\end{equation}
which proves the other identity in \eqref{au-selfconsistency}.
\end{proof}

\begin{remark}
These identities play the same role as the fixed-point identities for
optimized R\'enyi information in (public) classical communication and quantum communication
\cite[Lemmas~6 and~9]{CT}.
\end{remark}

\subsection{Additivity for fixed ensembles}\label{sec:ensemble-additivity}
The fixed-point identities also yield additivity under tensor products.
\begin{lemma}[Additivity for fixed ensembles]\label{lem:ensemble-additivity}
For every $\alpha>1$ and any two finite cqq states $\omega_{XBE}$ and
$\zeta_{YB'E'}$,
\begin{equation}
 P_\alpha^{(1)}(\omega\otimes\zeta)
 =P_\alpha^{(1)}(\omega)+P_\alpha^{(1)}(\zeta).
 \label{eq:ensemble-additivity}
\end{equation}
\end{lemma}
\begin{proof}[Proof of Lemma~\ref{lem:ensemble-additivity}]
First assume that all conditional marginals are faithful, and choose
optimizing tuples from Lemma~\ref{au-optimizer-theorem} for the two
ensembles. Denote their reference and auxiliary states by
$\tau_B,\sigma_E,\eta_{B,x},\nu_{E,x}$ and
$\tau_{B'},\sigma_{E'},\eta_{B',y},\nu_{E',y}$, their letter values by
$\ell_x,\ell'_y$, and their tilted probabilities by $r_x,r'_y$.
Consider the product references and the product conditional auxiliaries.
The conditional maps tensor, and the Schatten norms multiply. The
normalized $\beta$ powers in \eqref{eq:optimal-eve-state} therefore give
$\nu_{EE',xy}=\nu_{E,x}\otimes\nu_{E',y}$ and
$\ell_{xy}=\ell_x\ell'_y$. Moreover, the power-stationarity equations
\eqref{au-power-stationarity} tensor to give
\begin{equation}
 C_{BB',xy}
 =\ell_x\ell'_y(\eta_{B,x}\otimes\eta_{B',y})^{1-u},
 \label{alt:product-power-stationarity}
\end{equation}
where $C_{BB',xy}$ is the left side of
\eqref{au-power-stationarity} for the product ensemble.
At fixed references the dual quasi-norm letter objective is concave in its
auxiliary state. Its constrained differential vanishes at the product
state by \eqref{alt:product-power-stationarity}, so this state is a
global maximizer, including among nonproduct auxiliaries.

The two reference identities also tensor:
\begin{align}
 \sum_{x,y}r_xr'_y\nu_{E,x}\otimes\nu_{E',y}
 &=\sigma_E\otimes\sigma_{E'},
 \label{alt:product-eve-stationarity}\\
 \sum_{x,y}r_xr'_y\eta_{B,x}\otimes\eta_{B',y}
 &=\tau_B\otimes\tau_{B'}.
 \label{alt:product-bob-stationarity}
\end{align}
The differential calculations in the proof of
Lemma~\ref{au-optimizer-theorem} can be read in both directions. If $Z$ denotes the summed letter value,
these identities give the gradients $2Z\tau^{1+u/2}$ with respect to
$\tau^{-u/2}$ and $2Z\sigma^{1-u/2}$ with respect to
$\sigma^{u/2}$, respectively. Each is proportional to the gradient of
its trace-one constraint, so the constrained derivatives vanish.
After maximizing the conditional auxiliaries, the objective
is concave in Eve's reference by Lemma~\ref{lem:dual-norm};
thus \eqref{alt:product-eve-stationarity} gives its global maximum at
the product Eve reference. After this maximization, the objective is
convex in Bob's reference by Lemma~\ref{lem:dual-norm}; hence
\eqref{alt:product-bob-stationarity} gives its global minimum at the
product Bob reference. Differentiation of these optimized objectives
is legitimate: uniqueness, faithfulness, and the envelope argument
are established in Lemma~\ref{au-optimizer-theorem} and its proof.
The optimized exponential value is consequently
\begin{equation}
 \sum_{x,y}p_xq_y\ell_x\ell'_y
 =\left(\sum_xp_x\ell_x\right)
  \left(\sum_yq_y\ell'_y\right),
 \label{alt:product-optimal-value}
\end{equation}
which proves the claim for faithful conditional marginals.
Mix each conditional signal with its maximally mixed state and let the
mixing parameter tend to zero. Corollary~\ref{cor:state-continuity}
proves the general statement.
\end{proof}

The channel quantity still optimizes over ensembles of correlated block
inputs, so this fixed-ensemble identity does not imply channel additivity.

\subsection{The tilted state and integral representation}\label{sec:integral}
Fix a cqq state $\omega_{XBE}$ satisfying the faithfulness assumptions of
Lemma~\ref{au-optimizer-theorem}. For each $0<u<1$, choose a faithful
optimizing tuple $\tau_B^\star,\sigma_E^\star,
\boldsymbol\nu_E^\star,\boldsymbol\eta_B^\star$ for
\eqref{eq:optimizer-u-objective}. All these states, as well as
$\ell_x$ and $r_x$, depend on $u$; we suppress this dependence where
the parameter is fixed. Define the operator
\begin{equation}
 L_{BE,x}(u)
 :=(\eta_{B,x}^\star)^{u/2}(\tau_B^\star)^{-u/2}
       \otimes(\nu_{E,x}^\star)^{-u/2}(\sigma_E^\star)^{u/2}.
 \label{eq:tilt-filter}
\end{equation}
The positive operator multiplying $\omega_{BE,x}$ in
\eqref{eq:trace-objective} is exactly $L_{BE,x}^\dagger L_{BE,x}$.
Consequently, the optimized sum $\sum_xp_x\ell_x(u)$ normalizes
the following tilted ensemble:
\begin{equation}
 \Theta_{XBE}(u)
 :=\frac{1}{\sum_xp_x\ell_x(u)}\sum_xp_x\ket x\bra x\otimes
       L_{BE,x}(u)\omega_{BE,x}L_{BE,x}(u)^\dagger.
 \label{eq:tilted-ensemble}
\end{equation}
We call this the tilted state. Its letter probabilities are
$r_x(u)=p_x\ell_x(u)/\sum_yp_y\ell_y(u)$, since
$\ell_x(u)=\Tr L_{BE,x}(u)\omega_{BE,x}L_{BE,x}(u)^\dagger$.
Its normalized conditional states $\Theta_{BE,x}$ and its averages obey
\begin{align}
 \Theta_{B,x}(u)&=\eta_{B,x}^\star,
 &\Theta_{E,x}(u)&=\nu_{E,x}^\star,
 \label{eq:tilt-marginals}\\
 \Theta_B(u)&=\tau_B^\star,
 &\Theta_E(u)&=\sigma_E^\star.
 \label{eq:tilt-averages}
\end{align}
Indeed, the conditional stationarity identities
\eqref{eq:optimal-eve-state} and \eqref{au-power-stationarity} give
\begin{align}
 \Tr_E\Theta_{BE,x}
 &=\ell_x^{-1}(\eta_{B,x}^\star)^{u/2}
       \bigl[\ell_x(\eta_{B,x}^\star)^{1-u}\bigr]
       (\eta_{B,x}^\star)^{u/2}=\eta_{B,x}^\star,\\
 \Tr_B\Theta_{BE,x}
 &=\ell_x^{-1}(\nu_{E,x}^\star)^{-u/2}
       \bigl[\ell_x(\nu_{E,x}^\star)^{1+u}\bigr]
       (\nu_{E,x}^\star)^{-u/2}=\nu_{E,x}^\star.
\end{align}
Averaging with $(r_x)$ and using \eqref{au-selfconsistency} proves
\eqref{eq:tilt-averages}.

For the support estimates in Section~\ref{sec:capacity-proof}, we also
retain a purification $\ket{\psi_x}_{BES}$ of each original conditional
state. The same operator gives a conditional purification of the tilt:
\begin{align}
 \ket{\phi_x(u)}
 &:=\ell_x^{-1/2}(L_{BE,x}(u)\otimes I_S)\ket{\psi_x},
 \label{eq:tilt}\\
 \Theta_{XBES}(u)
 &:=\sum_xr_x\ket x\bra x\otimes
       \ket{\phi_x(u)}\bra{\phi_x(u)}.
 \label{eq:tilted-purification}
\end{align}
Tracing out $S$ recovers \eqref{eq:tilted-ensemble}.

\begin{proof}[Proof of Lemma~\ref{lem:integral}]
Lemma~\ref{au-optimizer-theorem} confines the minimizing Bob references
to a compact interior set when $u$ ranges over a compact subinterval
of $(0,1)$. On this set the objective after the inner maximization is
continuously differentiable. The derivative rule for optimized values
\cite{Danskin} therefore makes $u\mapsto\log\sum_xp_x\ell_x(u)$
locally Lipschitz. At every differentiability point we may hold all four
optimizing families fixed, including the unique inner minimizers
$\nu_{E,x}^\star$ in \eqref{eq:optimal-eve-state}.

Differentiate the trace terms in \eqref{eq:optimizer-u-objective}
with these optimizing states held fixed. Cyclicity of the trace and
the conditional stationarity identities \eqref{eq:optimal-eve-state}
and \eqref{au-power-stationarity} give
\begin{align}
 \partial_u\log\ell_x
 &=\Tr\eta_{B,x}^\star\log\eta_{B,x}^\star
   -\Tr\eta_{B,x}^\star\log\tau_B^\star
   -\Tr\nu_{E,x}^\star\log\nu_{E,x}^\star
   +\Tr\nu_{E,x}^\star\log\sigma_E^\star\notag\\
 &=D(\eta_{B,x}^\star\|\tau_B^\star)
          -D(\nu_{E,x}^\star\|\sigma_E^\star).
 \label{eq:tilt-trace-derivative}
\end{align}
Differentiating the optimized sum and applying
\eqref{au-selfconsistency} and \eqref{eq:tilt-marginals} yields
\begin{align}
 \frac{d}{du}\log\sum_xp_x\ell_x(u)
 &=\sum_xr_x(u)
   \bigl[D(\eta_{B,x}^\star\|\tau_B^\star)
          -D(\nu_{E,x}^\star\|\sigma_E^\star)\bigr]\notag\\
 &=I(X:B)_{\Theta(u)}-I(X:E)_{\Theta(u)}.
 \label{eq:integral-derivative}
\end{align}
This holds for almost every $u\in(0,1)$ and for any minimizing
$\tau_B^\star$ at such a point; no differentiable choice of optimizers
is required.

By Lemma~\ref{lem:ensemble-order-one} and \eqref{eq:optimized-value-u},
$\log\sum_xp_x\ell_x(u)\to0$ as $u\searrow0$.
The private information of a tilted cqq state lies between
$-\log d_E$ and $\log d_B$. Integrate
\eqref{eq:integral-derivative} from $\varepsilon$ to
$(\alpha-1)/\alpha$, let $\varepsilon\searrow0$, and use
\eqref{eq:optimized-value-u}. This proves \eqref{eq:intro-integral}.
\end{proof}

\section{Continuity of the regularized R\'enyi private information}\label{sec:capacity-proof}
We prove that the regularized R\'enyi private information converges to
private capacity as $\alpha\searrow1$.
The proof compares R\'enyi and ordinary private information uniformly
over all blocklengths. Lemma~\ref{lem:integral} reduces this comparison
to restoring tilted states to the channel image.

\begin{theorem}[Continuity at order one]\label{thm:quantitative-continuity}
Let $\mathcal W:A\to BE$ be a finite-dimensional channel, with
$d_B=\dim B$ and $d_E=\dim E$. There is a nonnegative function
$\Gamma_{d_B,d_E}$, depending only on these output dimensions, such that
for $1<\alpha\le2$,
\begin{equation}
 0\le P_\alpha(\mathcal W)-P(\mathcal W)
 \le\Gamma_{d_B,d_E}(\alpha),
 \label{eq:quantitative-continuity}
\end{equation}
where $\Gamma_{d_B,d_E}(\alpha)\to0$ as $\alpha\searrow1$.
\end{theorem}

\subsection{Local estimates for reversing the tilt}\label{sec:conditional-tools}
For $0<u\le1/2$, we estimate how much of the
tilted state $\Theta(u)$ from \eqref{eq:tilted-ensemble} lies outside
the Stinespring output subspace at any single channel use.
The estimates below control the local Bob and Eve corrections uniformly
in blocklength. The resulting information cost per restored use tends
to zero with $u$, which is enough to pass to the regularized quantity.

We first assume that each conditional marginal is faithful on the
support of its average; the final approximation removes this assumption. We use the following result from~\cite{CT}.

\begin{lemma}[Conditional powers {\cite[Lemma~11]{CT}}]\label{lem:conditional-power}
For a state $\rho_{CD}$ and $-1\le v\le1$,
\begin{equation}
 \left\|\bigl[(I_C\otimes\rho_D)^{-v/2}\rho_{CD}^{v/2}-I_{CD}\bigr]
 \sqrt{\rho_{CD}}\right\|_2
 \le\sqrt{2(d_C-1)}\,|v|.
 \label{eq:conditional-power}
\end{equation}
Also, for $0\le u\le1$,
\begin{equation}
 \rho_{CD}\le d_C I_C\otimes\rho_D,\qquad
 \norm{(I_C\otimes\rho_D)^{-u/2}\rho_{CD}^{u/2}}_\infty\le d_C^{u/2}.
 \label{eq:reduction}
\end{equation}
\end{lemma}

\paragraph{Eve's correction.}
The correction below cancels the inverse tilt on $CD$ while leaving an operator only on $D$. In the application $C=E_i$ and $D=E_{-i}$, so that remaining operator commutes with the local channel-support projector. Below, $C_d$ and $C_{d_B,d_E}$ denote constants depending only on the
indicated dimensions and may change between estimates.
We work with $0<u\le1/2$.

\begin{lemma}[Local cancellation]\label{lem:transport}
Let $(p_x)$ be a finite probability distribution, let
$\nu_{CD,x}\in\Sc(CD)$, put $\nu_{D,x}=\Tr_C\nu_{CD,x}$, and set $\sigma_{CD}=\sum_xp_x\nu_{CD,x}$.
For every $0<u\le1/2$, there are operators $K_{D,x}\in\mathcal L(D)$ such that
\begin{align}
 C_x=(I_C\otimes K_{D,x})\sigma_{CD}^{-u/2}(\nu_{CD,x})^{u/2}&,
 \label{eq:transport-C}\\
 \left(\sum_xp_x
 \norm{(C_x-I)(\nu_{CD,x})^{(1-u)/2}}_{2/(1-u)}^2\right)^{1/2}
 &\le C_{d_C}u.
 \label{eq:transport-high}
\end{align}
The bound is independent of the untouched dimension, ensemble size, and
spectra. For $d_C=1$ it is zero.
\end{lemma}
\begin{proof}
Define the operator $J_{CD}$ and the positive operators $Z_x$ on $D$ by
\begin{equation}
 J_{CD}=(I_C\otimes\sigma_D)^{u/2}\sigma_{CD}^{-u/2}-I,
 \qquad Z_x=\Tr_C(J_{CD}\nu_{CD,x}J_{CD}^\dagger).
\end{equation}
Using $\sum_xp_x\nu_{CD,x}=\sigma_{CD}$ and \eqref{eq:conditional-power} gives
\begin{equation}
 \sum_xp_x\Tr Z_x=\norm{J_{CD}\sqrt\sigma_{CD}}_2^2\le C_{d_C}u^2.
 \label{eq:average-reference-error}
\end{equation}
Choose
\begin{equation}
 \mu_{D,x}=\nu_{D,x}+Z_x,\qquad
 K_{D,x}=(\mu_{D,x})^{-u/2}\sigma_D^{u/2},
\end{equation}
and write
\begin{align}
C_x-I&=D_x+(I_C\otimes\mu_{D,x})^{-u/2}J_{CD}(\nu_{CD,x})^{u/2},\\
 D_x&=(I_C\otimes\mu_{D,x})^{-u/2}(\nu_{CD,x})^{u/2}-I.
\end{align}
With identities on $C$ understood, split
\begin{equation}
 D_x=(\mu_{D,x})^{-u/2}\bigl[(\nu_{CD,x})^{u/2}-(\nu_{D,x})^{u/2}\bigr]
       +\bigl[(\mu_{D,x})^{-u/2}(\nu_{D,x})^{u/2}-I\bigr].
\end{equation}
Since $\mu_{D,x}\ge\nu_{D,x}$, the operator
$(\mu_{D,x})^{-u/2}(\nu_{D,x})^{u/2}$ is a contraction. Equation~\eqref{eq:conditional-power}
bounds the first term in state-weighted Hilbert--Schmidt norm by $C_{d_C}u$.
For a state $A$ and $B\ge A$, the scalar inequality used in
\cite[Lemma~11]{CT},
$(e^{uy/2}-1)^2\le u^2(e^y+e^{-y}-2)$ for $|u|\le1$,
gives, by expansion in eigenbases of $A$ and $B$,
\begin{align}
\|[B^{-u/2}A^{u/2}-I]\sqrt A\|_2^2
 &\le u^2\bigl[\Tr A^2B^{-1}+\Tr\Pi_A B-2\bigr]\\
 &\le u^2\Tr(B-A).
\label{eq:ordered-power}
\end{align}
Here $\Pi_A$ is the support projector; $B\ge A$ implies
$\Tr A^2B^{-1}\le\Tr A=1$.
The second term in $D_x$ acts only on $D$, so this estimate with
$A=\nu_{D,x}$ and $B=\mu_{D,x}$ bounds its weighted norm by
$u\sqrt{\Tr Z_x}$. Minkowski's inequality therefore gives
\begin{equation}
 \left(\sum_xp_x\norm{D_x\sqrt{\nu_{CD,x}}}_2^2\right)^{1/2}
 \le C_{d_C}u+u\left(\sum_xp_x\Tr Z_x\right)^{1/2}\le(1+u)C_{d_C}u.
 \label{eq:D-L2}
\end{equation}
Reduction gives $\norm{D_x}_\infty\le1+d_C^{u/2}$. The weighted
interpolation estimate
\begin{equation}
 \norm{D_x(\nu_{CD,x})^{(1-u)/2}}_{2/(1-u)}
 \le\norm{D_x\sqrt{\nu_{CD,x}}}_2^{1-u}\norm{D_x}_\infty^u
\end{equation}
and concavity of $t^{1-u}$ imply
\begin{equation}
 \left(\sum_xp_x
 \norm{D_x(\nu_{CD,x})^{(1-u)/2}}_{2/(1-u)}^2\right)^{1/2}
 \le C_{d_C}u.
 \label{eq:D-high}
\end{equation}
Here the dimension-dependent factors are bounded for $0<u\le1/2$,
and $u^{-u}\le e^{1/e}$.
For the second summand in $C_x-I$, reduction gives
\begin{equation}
 J_{CD}\nu_{CD,x}J_{CD}^\dagger\le d_C I_C\otimes Z_x
       \le d_C I_C\otimes\mu_{D,x}.
\end{equation}
For clarity, put $A_x=J_{CD}\nu_{CD,x}J_{CD}^\dagger$ and
$M_x=I_C\otimes\mu_{D,x}$. The relation $A_x\le d_C M_x$ gives
$\|M_x^{-u/2}A_x^{u/2}\|_\infty\le d_C^{u/2}$ by operator
monotonicity. Polar decomposition and Schatten H\"older applied to
$M_x^{-u/2}A_x^{1/2}
=(M_x^{-u/2}A_x^{u/2})A_x^{(1-u)/2}$ therefore give
\begin{equation}
 \norm{(I_C\otimes\mu_{D,x})^{-u/2}J_{CD}\sqrt{\nu_{CD,x}}}_{2/(1-u)}
 \le d_C^{u/2}(\Tr Z_x)^{(1-u)/2}.
\end{equation}
Average using concavity of $t^{1-u}$, substitute
\eqref{eq:average-reference-error}, and use Minkowski's inequality
to obtain \eqref{eq:transport-high}. If $d_C=1$,
then $J_{CD}=0$ on the signal supports. All inverse powers are taken
on their supports, so the construction and estimates also apply to
singular states.
\end{proof}

\paragraph{Combining the two corrections.}
The following consequence of data processing controls weighted operator
norms under a channel adjoint. We first use it for a partial trace in
Lemma~\ref{lem:cross}, and then for Bob's correction in
Section~\ref{sec:localization}. 

\begin{lemma}[A norm inequality from data processing]\label{lem:dual-dpi}
For every quantum channel $\mathcal N$, every state $\rho$ on its input
space, every $p>1$, and every positive operator $Y$ on its output space,
\begin{equation}
 \|\rho^{1/(2p)}\mathcal N^\dagger(Y)\rho^{1/(2p)}\|_p
 \le\|\mathcal N(\rho)^{1/(2p)}Y\mathcal N(\rho)^{1/(2p)}\|_p.
 \label{eq:dual-dpi}
\end{equation}
\end{lemma}
\begin{proof}
Put $p'=p/(p-1)$ and $\sigma=\mathcal N(\rho)$, restricting all
operators to the relevant supports. For $Z\ge0$ put
$A=\rho^{1/(2p)}Z\rho^{1/(2p)}$.
Sandwiched R\'enyi data processing at order $p'$ for the channel
$\mathcal N$ \cite[Theorem~6]{Beigi} gives
\begin{equation}
 \|\sigma^{-1/(2p)}\mathcal N(A)\sigma^{-1/(2p)}\|_{p'}
 \le\|Z\|_{p'}.
\end{equation}
Homogeneity removes the normalization of $A$. H\"older's inequality yields
\begin{equation}
 \Tr Z\rho^{1/(2p)}\mathcal N^\dagger(Y)\rho^{1/(2p)}
 =\Tr Y\mathcal N(A)
 \le\|\sigma^{1/(2p)}Y\sigma^{1/(2p)}\|_p\|Z\|_{p'}.
\end{equation}
Taking the supremum over positive $Z$ with $\|Z\|_{p'}\le1$ proves
the claim by Schatten duality.
\end{proof}

\begin{lemma}[Combining the Bob and Eve corrections]\label{lem:cross}
For a pure $\ket\phi_{BES}$ with marginals $\eta_B,\nu_E$, arbitrary operators $A_B,D_E$, and $0<u<1$,
\begin{equation}
 \norm{(A_B\otimes D_E\otimes I_S)\phi}_2
 \le\norm{A_B\eta_B^{u/2}}_{2/u}
       \norm{D_E\nu_E^{(1-u)/2}}_{2/(1-u)}.
 \label{eq:cross}
\end{equation}
\end{lemma}
\begin{proof}
Schmidt decomposition across $B:ES$ and Schatten H\"older give the right side with the second factor replaced by
$\norm{(D_E\otimes I_S)\rho_{ES}^{(1-u)/2}}_{2/(1-u)}$,
where $\rho_{ES}=\Tr_B\ket\phi\bra\phi$. Equation~\eqref{eq:dual-dpi}, applied to $\Tr_S$ with $p=1/(1-u)$ and $Y=D_E^\dagger D_E$, bounds the square of this factor by
$\norm{\nu_E^{(1-u)/2}D_E^\dagger D_E\nu_E^{(1-u)/2}}_{1/(1-u)}$.
Taking the square root proves the claim. The constant is one, independently of the purifying dimension.
\end{proof}

\subsection{Projecting the tilt back to the channel image}\label{sec:localization}
Fix $n$ and $0<u\le1/2$, and take the tilted
state $\Theta_{XB^nE^nS}(u)$ from \eqref{eq:tilted-purification}
for an $n$-use ensemble. Choose a Stinespring isometry $V:A\to BFE$ with
$\mathcal W(\rho)=\Tr_F V\rho V^\dagger$, where $F$ is unobserved,
and include $F^n$ in the purifying register $S$. Keep $F^n$ until the support repairs are complete. Write
\begin{equation}
 P_i=VV^\dagger\quad\text{on }B_iF_iE_i,
 \qquad Q_i=I-P_i.
\end{equation}
The original purifications satisfy $P_i\ket{\psi_x}=\ket{\psi_x}$.
Here the optimizing reference states are $\tau_{B^n}^\star$ and
$\sigma_{E^n}^\star$. Write
$\eta_{B_{-i},x}^\star=\Tr_{B_i}\eta_{B^n,x}^\star$ and
$\tau_{B_{-i}}^\star=\Tr_{B_i}\tau_{B^n}^\star$ for their relevant
Bob marginals. Identities on untouched tensor factors are implicit.

\paragraph{Bob's correction.}
For the split $B_i:B_{-i}$ define, with identities on $B_i$ implicit,
\begin{equation}
 C_{B,i,x}=(\eta_{B_{-i},x}^\star)^{u/2}(\tau_{B_{-i}}^\star)^{-u/2}
 (\tau_{B^n}^\star)^{u/2}(\eta_{B^n,x}^\star)^{-u/2}.
 \label{eq:CB}
\end{equation}
The bounds we need are
\begin{align}
 \left(\sum_xr_x\norm{(C_{B,i,x}-I)\sqrt{\eta_{B^n,x}^\star}}_2^2\right)^{1/2}
 &\le C_{d_B}u,
 \label{eq:Bob-L2}\\
 \norm{C_{B,i,x}(\eta_{B^n,x}^\star)^{u/2}}_{2/u}\le1.
 \label{eq:Bob-high}
\end{align}
The second is directly \eqref{eq:dual-dpi} for the partial trace over $B_i$, with $\rho=\tau_{B^n}^\star$, $p=1/u$, and $Y=(\tau_{B_{-i}}^\star)^{-u/2}(\eta_{B_{-i},x}^\star)^{u}(\tau_{B_{-i}}^\star)^{-u/2}$:
\begin{equation}
 \norm{C_{B,i,x}(\eta_{B^n,x}^\star)^{u/2}}_{2/u}^2
 =\norm{(\tau_{B^n}^\star)^{u/2}(I_{B_i}\otimes Y)(\tau_{B^n}^\star)^{u/2}}_{1/u}
 \le\norm{(\eta_{B_{-i},x}^\star)^{u}}_{1/u}=1.
\end{equation}
For \eqref{eq:Bob-L2}, define the operator on $B^n$
\begin{equation}
 T_i:=\bigl[I_{B_i}\otimes(\tau_{B_{-i}}^\star)^{-u/2}\bigr]
       (\tau_{B^n}^\star)^{u/2}-I_{B^n}.
 \label{eq:Bob-reference-deviation}
\end{equation}
It depends only on the reference state $\tau_{B^n}^\star$ and is therefore
independent of the signal label $x$. We split
\begin{equation}
 C_{B,i,x}-I=
 [(\eta_{B_{-i},x}^\star)^{u/2}(\eta_{B^n,x}^\star)^{-u/2}-I]
 +(\eta_{B_{-i},x}^\star)^{u/2}T_i(\eta_{B^n,x}^\star)^{-u/2}.
\end{equation}
The first error has state-weighted norm at most $C_{d_B}u$ by Equation~\eqref{eq:conditional-power}. For the second, use the joint concavity of the generalized Petz moment
$(A,B)\mapsto\Tr A^{1-u}T_i^\dagger B^uT_i$
\cite[Theorem~1 and Corollary~1.1]{Lieb}. Self-consistency gives
\begin{align}
 \sum_xr_x\Tr(\eta_{B^n,x}^\star)^{1-u}T_i^\dagger(\eta_{B_{-i},x}^\star)^{u}T_i
 &\le\Tr(\tau_{B^n}^\star)^{1-u}T_i^\dagger(\tau_{B_{-i}}^\star)^{u}T_i\\
 &=\norm{[I-(\tau_{B_{-i}}^\star)^{u/2}(\tau_{B^n}^\star)^{-u/2}]\sqrt{\tau_{B^n}^\star}}_2^2
 \le C_{d_B}u^2.
\end{align}
Minkowski's inequality proves \eqref{eq:Bob-L2}.

\paragraph{Joint cancellation and the support test.}
Apply Lemma~\ref{lem:transport} with $p_x=r_x(u)$ to Eve's ensemble.
Its average equals the reference by \eqref{au-selfconsistency}:
\begin{equation}
 \sigma_{E^n}^\star=\Theta_{E^n}(u)=\sum_xr_x\nu_{E^n,x}^\star
\end{equation}
with $C=E_i$, $D=E_{-i}$. Let $C_{E,i,x}$ be the resulting operator. By construction it has the form
\begin{equation}
 C_{E,i,x}=(I_{E_i}\otimes K_{E_{-i},x})(\sigma_{E^n}^\star)^{-u/2}(\nu_{E^n,x}^\star)^{u/2},
\end{equation}
and therefore, using the definition of the tilted vector,
\begin{equation}
 Q_i(C_{B,i,x}\otimes C_{E,i,x}\otimes I_S)\ket{\phi_x}=0.
 \label{eq:support-identity}
\end{equation}
Indeed, after the four tilt factors cancel, the remaining Bob and Eve operators act only on the systems outside $B_iF_iE_i$ and hence commute with $P_i$.

Use the two-term identity
\begin{equation}
 I-C_{B,i,x}\otimes C_{E,i,x}
 =(I-C_{B,i,x})\otimes I
   +C_{B,i,x}\otimes(I-C_{E,i,x}).
\end{equation}
Together with \eqref{eq:support-identity}, it bounds $Q_i\phi_x$ by Bob's error plus Bob's full correction applied to Eve's error. Equation~\eqref{eq:Bob-L2} controls the first term. For the second, apply Lemma~\ref{lem:cross} with $A_B=C_{B,i,x}$ and $D_E=I-C_{E,i,x}$, followed by \eqref{eq:Bob-high} and \eqref{eq:transport-high}. Minkowski's inequality gives
\begin{equation}
 \sqrt{\Tr Q_i\Theta(u)}\le C_{d_B,d_E}u.
 \label{eq:local-support}
\end{equation}
Thus the probability that the support test $\{P_i,Q_i\}$ rejects the
tilted state is $O_{d_B,d_E}(u^2)$, uniformly in blocklength and in
the ensemble.

\subsection{Information cost of restoring the channel output}\label{sec:restoration}
We use the sharp continuity bound for conditional entropy in trace distance
\cite{BertaSharpContinuity,ChengLiuSharpContinuity}.

\begin{lemma}[Conditional entropy continuity~\cite{BertaSharpContinuity,ChengLiuSharpContinuity}]\label{lem:entropy-continuity}
Let $\rho_{CD},\zeta_{CD}$ be states with $d_C=\dim C\ge2$ and
$T(\rho_{CD},\zeta_{CD})\le t\le1-d_C^{-2}$. Then
\begin{equation}
 |H(C|D)_\rho-H(C|D)_\zeta|
 \le t\log(d_C^2-1)+h_2(t).
 \label{eq:entropy-continuity}
\end{equation}
If $d_C=1$, both conditional entropies are zero.
\end{lemma}

For states $\rho,\zeta$ on $XB^nE^n$ with
$T(\rho,\zeta)\le t\le1$, the following consequence controls the
change at one channel use:
\begin{align}
&\left|\bigl[I(X:B_i|B_{-i})_\rho-I(X:E_i|E_{-i})_\rho\bigr]
       -\bigl[I(X:B_i|B_{-i})_\zeta-I(X:E_i|E_{-i})_\zeta\bigr]\right|
 \label{eq:local-private-continuity-start}\\
 &\qquad\le C_{d_B,d_E}t\log\frac{e}{t}.
 \label{eq:local-private-continuity}
\end{align}
For $t\le1/2$, apply Lemma~\ref{lem:entropy-continuity} to the four
conditional entropies and use $h_2(t)\le t\log(e/t)$; terms with a
one-dimensional conditioned system vanish. For $t>1/2$, boundedness
of these conditional entropies in the local dimensions gives the same
estimate after increasing $C_{d_B,d_E}$. At $t=0$, the right side is
understood by continuity.
If the $XB_{-i}E_{-i}$ marginals agree, the chain rule identifies the
left-hand side with the change in global private information.

Fix a state $\zeta_i$ on $B_iF_iE_i$ supported on $P_i$ and define
\begin{equation}
 \mathcal R_i(\rho)=P_i\rho P_i+
 \zeta_i\otimes\Tr_{B_iF_iE_i}(Q_i\rho Q_i).
 \label{eq:repair-channel}
\end{equation}
This channel preserves the full untouched marginal, including $X$.

\begin{proposition}[Local restoration]\label{prop:restoration}
Let $\rho_{XB^nF^nE^n}$ be a state classical on $X$, and put $q=\Tr Q_i\rho$.
Then
\begin{align}
&I(X:B^n)_\rho-I(X:E^n)_\rho
  -I(X:B^n)_{\mathcal R_i(\rho)}+I(X:E^n)_{\mathcal R_i(\rho)}
 \label{eq:restoration-cost-start}\\
 &\qquad\le C_{d_B,d_E}\sqrt q\log\frac{e}{\sqrt q}.
 \label{eq:restoration-cost}
\end{align}
\end{proposition}
\begin{proof}
The projection estimate $\|\rho-P_i\rho P_i\|_1\le2\sqrt q$
follows by purifying $\rho$: the difference of the two rank-one
operators has trace norm $\sqrt{4q-3q^2}\le2\sqrt q$, and partial
trace contracts the trace norm. The replacement term in
\eqref{eq:repair-channel} is positive and has trace $q$, so the
triangle inequality gives
\begin{equation}
 T(\rho,\mathcal R_i(\rho))
 \le\tfrac12\|\rho-P_i\rho P_i\|_1+\tfrac q2
 \le\sqrt q+\tfrac q2\le\tfrac32\sqrt q.
 \label{eq:accepted-distance}
\end{equation}
Apply \eqref{eq:local-private-continuity-start}--\eqref{eq:local-private-continuity}
with $t=\min\{1,3\sqrt q/2\}$ and use marginal preservation.
Since $t\mapsto t\log(e/t)$ is concave and vanishes at zero,
$t\log(e/t)\le\tfrac32\sqrt q\log(e/\sqrt q)$.
This proves \eqref{eq:restoration-cost}, including $q=0$ by continuity.
\end{proof}

Sequential restoration channels act on disjoint output triples and
preserve each later rejected weight. Thus the conditional-entropy
telescoping argument of Leung--Smith \cite[Theorem~11]{LeungSmith},
together with \eqref{eq:local-support} and \eqref{eq:restoration-cost}, gives
\begin{equation}
 I(X:B^n)_{\Theta(u)}-I(X:E^n)_{\Theta(u)}
 \le P^{(1)}(\mathcal W^{\otimes n})
       +nC_{d_B,d_E}u\log\frac{e}{u}.
 \label{eq:private-tilted-bound}
\end{equation}
Indeed, each local cost is at most $C_{d_B,d_E}u\log(e/u)$ after
increasing the constant. The final restored state lies in the image
of $V^{\otimes n}$, so tracing out $F^n$ gives a valid cqq
channel-output ensemble.

To prove Theorem~\ref{thm:quantitative-continuity}, integrate
\eqref{eq:private-tilted-bound} using \eqref{eq:intro-integral}. This gives
\begin{equation}
 P_\alpha^{(1)}(\omega)
 \le P^{(1)}(\mathcal W^{\otimes n})
       +n\Gamma_{d_B,d_E}(\alpha),
       \qquad 1<\alpha\le2.
 \label{eq:radius-final}
\end{equation}
Here we may take
\begin{equation}
 \Gamma_{d_B,d_E}(\alpha)
 :=\frac{\alpha C_{d_B,d_E}}{\alpha-1}
 \int_0^{\frac{\alpha-1}{\alpha}}u\log\frac{e}{u}\,du.
 \label{eq:continuity-correction}
\end{equation}
The integrand extends continuously by zero at $u=0$, so
$\Gamma_{d_B,d_E}(\alpha)\to0$ as $\alpha\searrow1$.

For arbitrary physical ensembles, set
$\overline\omega=\sum_xp_x\omega_{BE,x}$ and replace each input state by
its mixture with the average input. The resulting output is
$\omega_{BE,x}(\varepsilon)=(1-\varepsilon)\omega_{BE,x}+
\varepsilon\overline\omega$ and has faithful conditional marginals
on the average supports. Apply the bound there, then use
Corollary~\ref{cor:state-continuity} to let $\varepsilon\searrow0$.
For each approximating ensemble, the cancellation is performed on
the common marginal supports. The local norm estimates use support
powers in the original tensor-product spaces and remain valid for
proper average supports. Thus \eqref{eq:radius-final} holds for all physical ensembles.

Finally, optimizing \eqref{eq:radius-final} over input ensembles and
blocklengths gives the upper bound in \eqref{eq:quantitative-continuity}.
For the reverse bound, the private coding theorem and the one-shot
converse in Theorem~\ref{thm:intro-one-shot} give
$R\le P_\alpha(\mathcal W)$ for every $R<P(\mathcal W)$; a one-message
code covers $P(\mathcal W)=0$. Thus $P(\mathcal W)\le P_\alpha(\mathcal W)$.
Since $\Gamma_{d_B,d_E}(\alpha)\to0$, this completes the proof of
Theorem~\ref{thm:quantitative-continuity} and
\eqref{eq:renyi-capacity-continuity}.

\section{Converse bounds}\label{sec:converse}

We prove an upper bound on fidelity and a lower bound on the joint
trace-distance error in terms of the R\'enyi private information of the
code's output ensemble. These bounds also control
$[p_{\rm succ}-\delta_T]_+$ from \eqref{eq:gap-definition} and provide
Theorem~\ref{thm:intro-one-shot} and the converse direction of
Theorem~\ref{thm:exact-quantum-exponent}.

Fix a $(1,M)$ secret-key transmission code $\mathcal C$ for a channel
$\mathcal W:A\to BE$, as defined in Section~\ref{sec:channel-exponent}.
Its cqq output ensemble before decoding is
\begin{equation}
 \omega_{MBE}=\frac1M\sum_{m=1}^M|m\rangle\langle m|
          \otimes\mathcal W(\rho_{A,m}).
 \label{eq:code-ensemble}
\end{equation}
For this ensemble, \eqref{eq:conditional-adjoint} defines
$\mathcal T_m^\dagger:\mathcal L(B)\to\mathcal L(E)$.
The operator $\mathcal T_m^\dagger(\Lambda_m)$ describes Eve's state on the
correct-decoding branch for message $m$.
We use $F(A,C)=\|\sqrt A\sqrt C\|_1^2$ also for positive,
possibly subnormalized operators. Applying the direct-sum formula for
fidelity to the classical registers in \eqref{eq:key-fidelity} gives
\begin{equation}
 F(\mathcal C;\mathcal W)=
 \max_{\sigma_E\in\Sc(E)}\left[\frac1M\sum_{m=1}^M
       \sqrt{F(\mathcal T_m^\dagger(\Lambda_m),\sigma_E)}\right]^2.
 \label{eq:fidelity-blocks}
\end{equation}
Thus \eqref{eq:fidelity-blocks} is an equivalent expression for the
criterion \eqref{eq:key-fidelity}, used here to prove the converse.

\begin{proposition}[One-shot fidelity converse]\label{prop:oneshot}
For every $(1,M)$ secret-key transmission code $\mathcal C$ with output
ensemble \eqref{eq:code-ensemble} and every $\alpha>1$, the fidelity in
\eqref{eq:key-fidelity} satisfies
\begin{equation}
 F(\mathcal C;\mathcal W)
 \le e^{-\frac{\alpha-1}{\alpha}[\log M-P_\alpha^{(1)}(\omega_{MBE})]}.
 \label{eq:one-shot-fidelity}
\end{equation}
\end{proposition}
\begin{proof}
Put $u=(\alpha-1)/\alpha$ and $\beta=1/(1+u)$. Fix $\tau_B\in\Sp(B)$ and
$\sigma_E\in\Sc(E)$, and set
\begin{equation}
 \gamma_{E,m}=\mathcal T_m^\dagger(\Lambda_m),\qquad
 a_m=\|\tau_B^{u/2}\Lambda_m\tau_B^{u/2}\|_{1/u},
\end{equation}
\begin{equation}
 k_m=\inf_{\nu\in\Sp(E)}
 \left\|\tau_B^{-u/2}\mathcal T_m
 (\sigma_E^{u/2}\nu^{-u}\sigma_E^{u/2})\tau_B^{-u/2}\right\|_\alpha.
\end{equation}
The average $M^{-1}\sum_m k_m$ is precisely the infimum over the
letterwise references in \eqref{eq:radius-norm}. Schatten duality and
the variational formula for the $\beta$ quasi-norm give the key estimate
\begin{align}
\|\sigma_E^{u/2}\gamma_{E,m}\sigma_E^{u/2}\|_\beta
 &=\inf_{\nu\in\Sp(E)}
   \Tr\Lambda_m\mathcal T_m(\sigma_E^{u/2}\nu^{-u}\sigma_E^{u/2})\\
 &\le a_m k_m.
\label{eq:decoder-norm-bound}
\end{align}
Araki--Lieb--Thirring \cite{Araki}, followed by
$0\le\Lambda_m\le I_B$, yields
\begin{align}
 a_m^{1/u}&\le\Tr\tau_B\Lambda_m^{1/u}\\
 &\le\Tr\tau_B\Lambda_m,
 \end{align}
 and hence, 
 \begin{align}
 \sum_m a_m^{1/u}&\le1,\\
 \sum_m a_m&\le M^{1-u}.
\end{align}
All of these inequalities include zero decoding effects.
Schatten H\"older with exponents $2/(1+u)$ and
$2/(1-u)$ gives
\begin{align}
 \sqrt{F(\gamma_{E,m},\sigma_E)}
 &\le\|\sqrt{\gamma_{E,m}}\sigma_E^{u/2}\|_{2/(1+u)}
       \|\sigma_E^{(1-u)/2}\|_{2/(1-u)}\\
 &=\|\sigma_E^{u/2}\gamma_{E,m}\sigma_E^{u/2}\|_\beta^{1/2}.
\end{align}
Consequently, \eqref{eq:decoder-norm-bound} and Cauchy--Schwarz imply
\begin{align}
 \left[\frac1M\sum_m\sqrt{F(\gamma_{E,m},\sigma_E)}\right]^2
 &\le\frac1{M^2}\left(\sum_m a_m\right)\left(\sum_m k_m\right)\\
 &\le M^{-u}\frac1M\sum_m k_m.
\end{align}
Maximize over $\sigma_E$ and minimize over $\tau_B$. Equations
\eqref{eq:fidelity-blocks} and \eqref{eq:radius-norm} give
\eqref{eq:one-shot-fidelity}.
\end{proof}

\begin{proof}[Proof of Theorem~\ref{thm:intro-one-shot}]
Apply Proposition~\ref{prop:oneshot} and use
$P_\alpha^{(1)}(\omega_{MBE})\le P_\alpha^{(1)}(\mathcal W)$.
\end{proof}

\begingroup\mdfsetup{nobreak=true}
\begin{proposition}[One-shot trace-distance converse]\label{prop:oneshot-trace}
For every $(1,M)$ secret-key transmission code $\mathcal C$ with output
ensemble \eqref{eq:code-ensemble} and every $\alpha>1$, the quantities
$p_{\rm succ},\delta_T$ in \eqref{eq:gap-definition} satisfy
\begin{align}
[p_{\rm succ}-\delta_T]_+
 &\le 1-\min_{\sigma_E\in\Sc(E)}
 T(\rho_{M\widehat M E},\kappa_{M\widehat M}\otimes\sigma_E)\label{eq:one-shot-gap}\\
 &\le e^{-\frac{\alpha-1}{2\alpha-1}[\log M-P_\alpha^{(1)}(\omega_{MBE})]}.
\label{eq:one-shot-trace}
\end{align}
Moreover,
$[p_{\rm succ}-\delta_T]_+^2\le F(\mathcal C;\mathcal W)\le p_{\rm succ}$.
\end{proposition}
\endgroup
\begin{proof}
Use the notation $u,\beta,\gamma_{E,m},a_m,k_m$ from the proof of
Proposition~\ref{prop:oneshot}, for arbitrary
$\tau_B\in\Sp(B)$ and $\sigma_E\in\Sc(E)$.
In particular, \eqref{eq:decoder-norm-bound} and
$\sum_m a_m^{1/u}\le1$ hold.

Block diagonality in $M,\widehat M$ and the
variational formula for the trace norm give
\begin{align}
&1-T(\rho_{M\widehat M E},\kappa_{M\widehat M}\otimes\sigma_E)\\
 &\quad=\frac1M\sum_m\min_{0\le D_m\le I_E}
       \{\Tr D_m\gamma_{E,m}+\Tr(I_E-D_m)\sigma_E\}.
\end{align}
The quantum Chernoff inequality \cite[Theorem~1]{Audenaert},
Araki--Lieb--Thirring, and \eqref{eq:decoder-norm-bound} bound each
minimum by
\begin{equation}
 \Tr\gamma_{E,m}^\beta\sigma_E^{1-\beta}
 \le\|\sigma_E^{u/2}\gamma_{E,m}\sigma_E^{u/2}\|_\beta^\beta
 \le(a_m k_m)^\beta.
\end{equation}
Scalar H\"older with exponents $(1+u)/u$ and $1+u$ now gives
\begin{align}
&1-T(\rho_{M\widehat M E},\kappa_{M\widehat M}\otimes\sigma_E)\\
 &\quad\le\frac1M
       \left(\sum_m a_m^{1/u}\right)^{u/(1+u)}
       \left(\sum_m k_m\right)^{1/(1+u)}\\
 &\quad\le M^{-u/(1+u)}
       \left[\frac1M\sum_m k_m\right]^{1/(1+u)}.
\end{align}
The same reference optimizations in \eqref{eq:radius-norm} prove
the exponential bound \eqref{eq:one-shot-trace}.

To compare with $[p_{\rm succ}-\delta_T]_+$, replace $\widehat M$ by a
perfect copy of $M$. This changes the decoded state by trace distance
$1-p_{\rm succ}$, while its distance from
$\kappa_{M\widehat M}\otimes\sigma_E$ becomes
$T(\rho_{ME},\pi_M\otimes\sigma_E)$. The triangle inequality gives
\begin{equation}
 T(\rho_{M\widehat M E},\kappa_{M\widehat M}\otimes\sigma_E)
 \le 1-p_{\rm succ}+T(\rho_{ME},\pi_M\otimes\sigma_E).
\end{equation}
Minimizing over $\sigma_E$ and using \eqref{eq:gap-definition} proves
the first trace-distance inequality. This is the argument of
\cite[Proposition~28]{WTB}, applied to each reference state before
optimization. Finally, $1-\sqrt F\le T$ gives
$[p_{\rm succ}-\delta_T]_+^2\le F(\mathcal C;\mathcal W)$, and
fidelity data processing under the test $M=\widehat M$ gives
$F(\mathcal C;\mathcal W)\le p_{\rm succ}$.
\end{proof}

%Applying Proposition~\ref{prop:oneshot-trace} to $\mathcal W^{\otimes n}$
%gives the corresponding trace-distance bounds for every $(n,M)$ code.

\subsection{Exponential strong converse}\label{sec:strong-converse}
Combining the one-shot converses with continuity at order one gives
exponential decay above private capacity.

\begin{corollary}[Exponential strong converse]\label{cor:quantitative-strong-converse}
For every finite-dimensional channel $\mathcal W:A\to BE$,
\begin{align}
 E_{\mathrm{sc}}(\mathcal W,R)&=0,
       &&0\le R\le P(\mathcal W),\\
 E_{\mathrm{sc}}(\mathcal W,R)&>0,
       &&R>P(\mathcal W).
 \label{eq:strong-converse-property}
\end{align}
For every $R>P(\mathcal W)$, there is a constant $c>0$ such that
all $(n,M)$ codes with $\log M\ge nR$ also satisfy
\begin{equation}
 [p_{\rm succ}-\delta_T]_+\le e^{-nc}.
 \label{eq:gap-strong-converse}
\end{equation}
\end{corollary}
\begin{proof}
By Theorem~\ref{thm:quantitative-continuity}, for $R>P(\mathcal W)$
we may choose a fixed $\alpha>1$ such that $P_\alpha(\mathcal W)<R$.
Propositions~\ref{prop:oneshot} and~\ref{prop:oneshot-trace}, applied
to $\mathcal W^{\otimes n}$, then give exponential decay of fidelity
and of the reliability--secrecy gap, with a positive decay rate
independent of $n$.

The private coding theorem gives $E_{\mathrm{sc}}(\mathcal W,R)=0$
for $R<P(\mathcal W)$. To a code of rate $r<P(\mathcal W)$, append
an independent label of size $\lceil e^{n(R-r)}\rceil$, ignored by the
encoder and guessed uniformly by Bob. Equation~\eqref{eq:fidelity-blocks}
shows that this divides fidelity by the label size, so
$E_{\mathrm{sc}}(\mathcal W,R)\le R-r$ for $R\ge r$.
Letting $r\uparrow P(\mathcal W)$ proves the zero exponent at capacity.
When $P(\mathcal W)=0$, a one-message code gives that endpoint directly.
\end{proof}

\section{Achievability of the strong converse exponent}\label{sec:coding}

We prove the upper bound on $E_{\mathrm{sc}}$ that matches the
converse in Section~\ref{sec:converse}. The argument has four steps:
\begin{enumerate}
\item Achieve the variational exponent. Use ordinary private coding
on a modified channel, then transfer its performance to the physical
channel with a relative-entropy penalty. Optimizing the modification
gives the bound in terms of $\widetilde P_\alpha^{(1)}$
(Proposition~\ref{alt:change-channel}).
\item Compare the two quantities. On tensor powers, permutation
symmetry allows fixed universal states to replace the optimized Bob
and Eve reference states. Local pinching then bounds the remaining
difference between the trace and variational formulations by a cost
that vanishes per channel use.
\item Transfer the coding construction. Some pinches depend on the
input word, so codes for the pinched outputs do not automatically
apply to the original channel. Invariant acceptance tests allow this
transfer and yield the exponent in terms of $P_\alpha^{(1)}$
(Proposition~\ref{alt:operational-bridge}).
\item Regularize by blocking. Treat arbitrary input blocks as single
channel uses and optimize over their ensembles and lengths. This
gives the exponent in terms of $P_\alpha(\mathcal W)$
(Theorem~\ref{shell-interpolation}).
\end{enumerate}
The second and third steps are the main comparison and coding arguments;
the final step uses blocking and the convexity of the operational exponent.
We first record the elementary operational properties needed below.
\begin{lemma}[Existence and convexity of the strong converse exponent]\label{lem:operational-convexity}
For a finite-dimensional channel $\mathcal W$, the liminf in
\eqref{eq:operational-exponent} is a limit for every $R\ge0$.
As a function of $R\in[0,\infty)$, $E_{\mathrm{sc}}(\mathcal W,R)$ is
convex, nondecreasing, and $1$-Lipschitz. Moreover,
\begin{equation}
 0\le E_{\mathrm{sc}}(\mathcal W,R)\le R,\qquad
 E_{\mathrm{sc}}(\mathcal W,R)\ge R-\log d_B.
\end{equation}
\end{lemma}
\begin{proof}
For the optimal fidelity in \eqref{eq:optimal-fidelity}, product codes give
\begin{equation}
 F^\star(n+m,R;\mathcal W)\ge
 F^\star(n,R;\mathcal W)F^\star(m,R;\mathcal W).
\end{equation}
Fekete's lemma therefore proves existence of the limit. A constant
encoder with uniform guessing has fidelity $1/\lceil e^{nR}\rceil$,
which gives $0\le E_{\mathrm{sc}}(\mathcal W,R)\le R$.
The bound $F(\mathcal C;\mathcal W)\le p_{\rm succ}\le d_B^n/M$
gives $E_{\mathrm{sc}}(\mathcal W,R)\ge R-\log d_B$.

For $R'\ge R$, append an independent label of size
$\lceil e^{n(R'-R)}\rceil$, ignored by the encoder and guessed uniformly
by Bob. This divides the fidelity in \eqref{eq:fidelity-blocks} by that
label size. Together with monotonicity in the rate, it gives
\begin{equation}
 0\le E_{\mathrm{sc}}(\mathcal W,R')-E_{\mathrm{sc}}(\mathcal W,R)
 \le R'-R.
\end{equation}
Concatenating codes at different rates proves convexity, with continuity
covering arbitrary rate mixtures.
\end{proof}

Grouping channel uses into blocks also gives
\begin{equation}
 E_{\mathrm{sc}}(\mathcal W^{\otimes k},R)
 =kE_{\mathrm{sc}}(\mathcal W,R/k).
 \label{eq:exponent-block-scaling}
\end{equation}

\subsection{Achievability of the variational exponent}\label{sec:variational-achievability}

The first step achieves the exponent expressed through the variational
quantity in Definition~\ref{alt:variational-definition}.

\begin{proposition}[Achievability of the variational exponent]
\label{alt:change-channel}
For every finite cqq output ensemble $\omega_{XBE}$ of a channel
$\mathcal W:A\to BE$ and every $R\ge0$,
\begin{equation}
 E_{\mathrm{sc}}(\mathcal W,R)
 \le\sup_{\alpha>1}\frac{\alpha-1}{\alpha}
       \left[R-\widetilde P_\alpha^{(1)}(\omega)\right].
 \label{alt:variational-achievable}
\end{equation}
\end{proposition}

The proof uses the following relative-entropy bound on fidelity.

\begin{lemma}[Relative-entropy fidelity bound]\label{lem:relative-entropy-fidelity}
Let $\rho\ll\zeta$ be states and let $\mathcal A$ be a completely
positive, trace-nonincreasing map with output registers $ME$, where
$M$ is classical. Put $a=\Tr\mathcal A(\rho)>0$ and
$\xi_{ME}=\mathcal A(\rho)/a$, and let $\pi_M$ be uniform on $M$.
Then
\begin{equation}
 \log F(\mathcal A(\zeta),\pi_M\otimes\xi_E)
 \ge H(M|E)_\xi-\log|M|
             -\frac{D(\rho\|\zeta)+h_2(a)}a.
 \label{eq:relative-entropy-fidelity}
\end{equation}
\end{lemma}
\begin{proof}
For a state $\xi$ and positive operators $A,B$, Gibbs' variational
principle and Golden--Thompson give
\begin{align}
 e^{-[D(\xi\|A)+D(\xi\|B)]/2}
 &\le\Tr e^{(\log A+\log B)/2},\\
 &\le\Tr\sqrt A\sqrt B,\\
 &\le\sqrt{F(A,B)}.
 \label{eq:fidelity-relative-entropy-variational}
\end{align}
Here $D(\xi\|A)=\Tr\xi(\log\xi-\log A)$ also for subnormalized
$A$; singular operators follow by approximation.
Complete $\mathcal A$ with a failure flag and write
$b=\Tr\mathcal A(\zeta)$. Relative-entropy data processing yields
\begin{align}
 D(\rho\|\zeta)
 &\ge aD(\xi\|\mathcal A(\zeta))
       +a\log a+(1-a)\log\frac{1-a}{1-b},\\
 &\ge aD(\xi\|\mathcal A(\zeta))-h_2(a).
\end{align}
Apply \eqref{eq:fidelity-relative-entropy-variational} with
$A=\mathcal A(\zeta)$ and $B=\pi_M\otimes\xi_E$, and use
$D(\xi\|B)=\log|M|-H(M|E)_\xi$.
\end{proof}

\begin{proof}[Proof of Proposition~\ref{alt:change-channel}]
Fix $\theta_{XBE}\ll\omega_{XBE}$, write its letter probabilities
as $q_x$, and put
\begin{align}
 I_\theta&:=I(X:B)_\theta-I(X:E)_\theta,\label{alt:aux-private}\\
 d_\theta&:=\sum_xq_xD(\theta_{BE,x}\|\omega_{BE,x}).\label{alt:aux-cost}
\end{align}
We first show the change-of-channel bound
\begin{equation}
 E_{\mathrm{sc}}(\mathcal W,R)\le d_\theta+[R-I_\theta]_+.
 \label{alt:change-channel-exponent}
\end{equation}
Suppose $I_\theta>0$, and fix $0<r<I_\theta$ and $\delta>0$.
Apply the ordinary private coding theorem to the modified channel
$x\mapsto\theta_{BE,x}$ with distribution $q$. Its typical-input
construction \cite[Theorem~1 and its direct proof]{Devetak} gives
codes with $M_k=\lceil e^{kr}\rceil$, vanishing decoding error and
mutual-information leakage, using only $q$-typical input words. Choose the typicality
tolerance so that every encoded word satisfies
\begin{equation}
 \sum_{i=1}^kD(\theta_{BE,x_i}\|\omega_{BE,x_i})
 \le k(d_\theta+\delta).
 \label{alt:typical-divergence-cost}
\end{equation}
Let $\rho_k^\theta$ and $\rho_k^\omega$ be the states of the uniform
message and the channel outputs, with the same stochastic encoder
applied to the modified and physical channels. Additivity and joint
convexity of relative entropy give
\begin{equation}
 D(\rho_k^\theta\|\rho_k^\omega)\le k(d_\theta+\delta).
 \label{alt:code-relative-entropy}
\end{equation}

Let $\mathcal A_k$ apply Bob's decoder, retain correct decoding,
and discard his output copy of the message. On the modified channel,
$a_k=\Tr\mathcal A_k(\rho_k^\theta)\to1$ and
$I(M:E^k)_{\rho_k^\theta}\to0$. Write $\xi_k$ for the normalized
accepted state. Conditioning on the binary correct-decoding flag
and using the conditional-entropy chain rule gives
\begin{align}
 D(\xi_k\|\pi_{M_k}\otimes(\xi_k)_{E^k})
 &=\log M_k-H(M|E^k)_{\xi_k},\\
 &\le\frac{I(M:E^k)_{\rho_k^\theta}+h_2(a_k)}{a_k}
 \longrightarrow0.
 \label{alt:code-entropy-deficit}
\end{align}
The fidelity of the physical code is at least the fidelity of its
correct-decoding part with $\pi_{M_k}\otimes(\xi_k)_{E^k}$.
Lemma~\ref{lem:relative-entropy-fidelity} therefore gives
\begin{equation}
 -\frac1k\log F(\mathcal C_k;\mathcal W)
 \le\frac{d_\theta+\delta+h_2(a_k)/k}{a_k}+o(1).
 \label{alt:change-channel-fidelity}
\end{equation}
Take $k\to\infty$ and then $\delta\searrow0$. Thus
$E_{\mathrm{sc}}(\mathcal W,r)\le d_\theta$ for $0<r<I_\theta$.
Monotonicity, continuity and the $1$-Lipschitz bound in
Lemma~\ref{lem:operational-convexity} extend this to
\eqref{alt:change-channel-exponent}. If $I_\theta\le0$, that bound
already follows from $E_{\mathrm{sc}}(\mathcal W,R)\le R$.

It remains to optimize over $\theta$. The relative-entropy chain rule gives
\begin{equation}
 D(\theta\|\omega)=D(q\|p)+d_\theta\ge d_\theta.
\end{equation}
For $0\le u\le1$, $uI_\theta-D(\theta\|\omega)$ is concave in
$\theta$: apart from the linear term $\Tr\theta\log\omega$, it equals
\begin{equation}
 (1-u)H(XBE)_\theta+uH(E|XB)_\theta
       +uH(X|E)_\theta+uH(B)_\theta.
 \label{eq:variational-objective-concavity}
\end{equation}
The supported cqq state space is compact, so Sion's minimax theorem
\cite{Sion} gives
\begin{align}
 \inf_{\theta\ll\omega}\{D(\theta\|\omega)+[R-I_\theta]_+\}
 &=\sup_{0\le u\le1}\inf_{\theta\ll\omega}
       \{D(\theta\|\omega)+u(R-I_\theta)\},\\
 &=\sup_{\alpha>1}\frac{\alpha-1}{\alpha}
       \left[R-\widetilde P_\alpha^{(1)}(\omega)\right].
 \label{alt:variational-exponent-minimax}
\end{align}
In the last line, set $u=(\alpha-1)/\alpha$; compactness and continuity
allow the endpoint values to be recovered as limits from $0<u<1$.
Together with \eqref{alt:change-channel-exponent}, this proves the claim.
\end{proof}

\subsection{Comparison by local pinching}\label{sec:pinching-comparison}

We now compare the two quantities. We first replace their reference
states by common universal states. These dominate every permutation-invariant state up to a
polynomial factor, which becomes an $O(\log n)$ cost after taking
logarithms. The conditional auxiliary optimizations remain.
We then use the log-Euclidean pinching approach \cite{MosonyiOgawa}
on both Bob's and Eve's systems to control their remaining difference.
Section~\ref{sec:pinching-transfer} then transfers the coding tests to
the original channel.

\begingroup\mdfsetup{nobreak=true}
\begin{proposition}[Asymptotic comparison by local pinching]
\label{alt:controlled-comparison}
For every finite cqq state $\omega_{XBE}$, there are ensembles
$\widehat\omega_n$ obtained by locally pinching the conditional
outputs of $\omega^{\otimes n}$, independently of $\alpha$, such
that for every $\alpha>1$,
\begin{equation}
 \lim_{n\to\infty}\frac1n
       \widetilde P_\alpha^{(1)}(\widehat\omega_n)
 =P_\alpha^{(1)}(\omega).
 \label{alt:asymptotic-comparison}
\end{equation}
\end{proposition}
\endgroup
The universal-reference lemma below fixes the outer references. We
then construct the local pinchings and prove the comparison.

\subsubsection{Universal reference states}\label{alt:foundations}
Fix $u=(\alpha-1)/\alpha$ and $\beta=1/(1+u)$. We temporarily fix
faithful references and write
\begin{align}
 J(\omega\mid\tau,\sigma)
 &:=\frac1u\log\sup_{\boldsymbol\eta_B}\inf_{\boldsymbol\nu_E}
 G_\alpha(\tau,\sigma,\boldsymbol\nu_E,\boldsymbol\eta_B),
 \label{alt:fixed-ordered-log}\\
 \widetilde J(\omega\mid\tau,\sigma)
 &:=\sup_{\theta\ll\omega}
 \left\{I(X:B)_\theta-I(X:E)_\theta-\frac1uD(\theta\|\omega)\right.
 \label{alt:fixed-variational-start}\\
 &\hspace{25mm}\left.+D(\theta_B\|\tau)-D(\theta_E\|\sigma)\right\}.
 \label{alt:fixed-games}
\end{align}
The second supremum is over cqq states. In $J$, only the outer
references $\tau_B,\sigma_E$ are fixed; the conditional auxiliaries
$\boldsymbol\eta_B,\boldsymbol\nu_E$ remain optimized. In
$\widetilde J$, the extra relative entropies account for replacing
$\theta_B,\theta_E$ by these same references. We can therefore
compare the two formulations at a common pair of reference states.

\begin{lemma}[Universal reference states]
\label{alt:outer-reference-lemma}\label{alt:universal-references}
For $R=B,E$, there are faithful permutation-invariant states
$\Omega_{R,n}$ with polynomially many eigenvalues such that every
permutation-invariant state satisfies
$\rho_{R^n}\le c_{R,n}\Omega_{R,n}$, where $c_{R,n}$ is polynomial
in $n$. For any permutation-covariant ensemble on $XB^nE^n$,
\begin{align}
 P_\alpha^{(1)}(\omega)
 &=J(\omega\mid\Omega_{B,n},\Omega_{E,n})+O(\log n),
 \label{alt:universal-ordered-lower}\\
 \widetilde P_\alpha^{(1)}(\omega)
 &=\widetilde J(\omega\mid\Omega_{B,n},\Omega_{E,n})+O(\log n).
 \label{alt:universal-variational-lower}
\end{align}
The errors are uniform in the ensemble and in $\alpha>1$ at fixed output dimensions.
\end{lemma}
\begin{proof}
Choose the standard universal states obtained by averaging normalized
projectors onto the Schur--Weyl summands \cite{MosonyiOgawa}.
The number of summands and their multiplicity dimensions are
polynomial, giving the stated spectral and domination bounds.

For the ordered quantity, Lemma~\ref{lem:dual-norm} permits averaging
Bob's minimizing reference and Eve's maximizing reference over
permutations. Universal domination and operator monotonicity of
$t^{-u}$ in the norm representation, and of $t^u$ in the dual
representation, place $P_\alpha^{(1)}-J$ between
$-\log c_{B,n}$ and $\log c_{E,n}$.

For the variational quantity, the objectives defining
$\widetilde P_\alpha^{(1)}$ and $\widetilde J$ are concave in $\theta$.
After multiplication by $u$, their common nonlinear terms are
$(1-u)H(XBE)+uH(E|XB)+uH(X|E)$, supplemented respectively by
$uH(B)$ and $-uD(\theta_E\|\sigma)$; the remaining terms are linear.
Covariance therefore permits averaging $\theta$ over permutations.
For invariant marginals, universal domination gives
\begin{equation}
 0\le D(\theta_{R^n}\|\Omega_{R,n})\le\log c_{R,n},\qquad R=B,E.
 \label{alt:universal-relative-entropy}
\end{equation}
The two variational objectives differ by the Bob relative entropy
minus the Eve relative entropy. Their suprema consequently differ
by at most $\log c_{B,n}+\log c_{E,n}$.
\end{proof}

\subsubsection{Comparison after controlled local pinching}
\label{alt:schur-comparison}
Let $\mathcal E_{R,n}$ be the spectral pinching of $\Omega_{R,n}$,
with $v_{R,n}$ projections. For a word $z\in\mathcal X^n$, its
stabilizer $G_z\cong\prod_x S_{n_x(z)}$ has local decompositions
\begin{equation}
 R^{\otimes n}=\bigoplus_\lambda
 \mathcal U_{R,z,\lambda}\otimes\mathcal V_{R,z,\lambda},
 \qquad R=B,E.
 \label{alt:schur-decomposition}
\end{equation}
The spaces $\mathcal V_{R,z,\lambda}$ carry inequivalent irreducible
representations; $\mathcal U_{R,z,\lambda}$ are the multiplicity
spaces. Let $\mathcal C_{R,z}$ pinch onto these summands, and define
\begin{align}
 \omega'_{B^nE^n,z}
 &:=(\mathcal E_{B,n}\otimes\mathcal E_{E,n})
       \left(\bigotimes_{i=1}^n\omega_{BE,z_i}\right),\\
 \widehat\omega_{B^nE^n,z}
 &:=(\mathcal C_{B,z}\otimes\mathcal C_{E,z})(\omega'_{B^nE^n,z}),\\
 \widehat\omega_n
 &:=\sum_zp^{\otimes n}(z)|z\rangle\langle z|
       \otimes\widehat\omega_{B^nE^n,z}.
 \label{alt:controlled-pinches}
\end{align}
We suppress output subscripts on the conditional states below.
The common spectral pinches commute with the word-dependent center
pinches, since the references commute with every permutation.
The conditional states are simultaneously $G_z$-invariant, and
$\widehat\omega_n$ remains permutation covariant. Schur--Weyl duality
gives polynomial bounds, uniform over words, on both the number
$V_{R,n}$ of projections in the composed pinching and
\begin{equation}
 D_{R,n}:=\max_z D_{R,z},\qquad
 D_{R,z}:=\sum_\lambda\dim\mathcal U_{R,z,\lambda}.
 \label{alt:effective-dimensions}
\end{equation}
The alphabet and local output dimensions are fixed throughout.

\begin{proof}[Proof of Proposition~\ref{alt:controlled-comparison}]
Put $\tau=\Omega_{B,n}$ and $\sigma=\Omega_{E,n}$.
A pinching with $v$ projections contracts the $\alpha$ norm and,
by the pinching inequality, reduces it by at most a factor $v$.
For the $\beta$ quasi-norm, concavity gives
$\|\mathcal E(A)\|_\beta\ge\|A\|_\beta$, while subadditivity of
$\Tr A^\beta$ gives $\|\mathcal E(A)\|_\beta\le v^u\|A\|_\beta$.
Apply these bounds to Bob's norm and Eve's dual quasi-norm
representations, respectively. Since the pinches fix the references,
\begin{equation}
 \left|J(\widehat\omega_n\mid\tau,\sigma)
          -J(\omega^{\otimes n}\mid\tau,\sigma)\right|
 \le\frac1u\log V_{B,n}+\log V_{E,n}.
 \label{alt:fixed-pinching-cost}
\end{equation}

To compare $J$ and $\widetilde J$ on $\widehat\omega_n$, average
both ordered auxiliaries and each variational conditional state over
$G_z$. This preserves optimality by concavity in $\eta$ after
minimization, convexity in $\nu$, and concavity of the variational
objective. For an invariant state
$A_R=\bigoplus_\lambda A_{R,\lambda}\otimes I_{\mathcal V_{R,z,\lambda}}$,
write $\overline A_R=\bigoplus_\lambda d_{R,\lambda}A_{R,\lambda}$,
where $d_{R,\lambda}=\dim\mathcal V_{R,z,\lambda}$.
This is a state in dimension $D_{R,z}$; a superscript $\uparrow$
lifts its blocks by tensoring with $I_{\mathcal V_{R,z,\lambda}}$.
On each block, the identity
\[
 \log A_{R,\lambda}=\log\overline A_{R,\lambda}
                    -(\log d_{R,\lambda})I
\]
holds on the support. Thus the $\log d_{R,\lambda}$ terms cancel in
relative entropies against the references; the scalar powers cancel
in the ordered products. Define
\begin{align}
 K_z&:=-\log\overline\tau^{\uparrow}\otimes I_{E^n}
             +I_{B^n}\otimes\log\overline\sigma^{\uparrow},\\
 Z&:=\sum_zp^{\otimes n}(z)\Tr\widehat\omega_z e^{uK_z}.
 \label{alt:common-partition-function}
\end{align}
The bars and lifts refer to the decomposition for $z$.
The common and center pinchings ensure that $K_z$ commutes with
$\widehat\omega_z$.

Suppressing the fixed arguments $(\widehat\omega_n\mid\tau,\sigma)$,
both quantities lie in the same interval:
\begin{equation}
 J,\;\widetilde J\in\left[\frac1u\log Z-\log D_{B,n},
          \frac1u\log Z+\log D_{E,n}\right].
 \label{alt:common-dimension-bound}
\end{equation}
For $J$, choose $\overline\eta_z=I/D_{B,z}$ for the lower bound
and $\overline\nu_z=I/D_{E,z}$ for the upper bound, using
$\overline\eta_z^u\le I$ and $\overline\nu_z^{-u}\ge I$.
For $\widetilde J$, its objective at
$\theta=\sum_zq_z|z\rangle\langle z|\otimes\theta_z$ is
\begin{equation}
 \sum_zq_z\bigl[-H(\overline\theta_{B,z})
                   +H(\overline\theta_{E,z})+\Tr\theta_zK_z\bigr]
                   -\frac1uD(\theta\|\widehat\omega_n).
 \label{alt:effective-entropy-objective}
\end{equation}
The entropy contribution lies between $-\log D_{B,n}$ and
$\log D_{E,n}$. Gibbs' formula maximizes the remaining terms to
$u^{-1}\log Z$; its optimizer has blocks proportional to
$p^{\otimes n}(z)\widehat\omega_z e^{uK_z}$ and retains the required
invariance. Hence $|\widetilde J-J|\le\log D_{B,n}+\log D_{E,n}$.
Combine this with \eqref{alt:fixed-pinching-cost}, the universal
reference bounds, and additivity
$P_\alpha^{(1)}(\omega^{\otimes n})=nP_\alpha^{(1)}(\omega)$.
All polynomial factors are independent of $\alpha$, so
\begin{equation}
 \sup_{\alpha>1}\frac{\alpha-1}{\alpha}
 \left|\widetilde P_\alpha^{(1)}(\widehat\omega_n)
              -nP_\alpha^{(1)}(\omega)\right|
 \le C\log(n+1).
 \label{alt:uniform-weighted-comparison}
\end{equation}
For fixed $\alpha>1$, division by $n$ proves the proposition.
The uniform weighted estimate also permits optimization over
$\alpha$ in the coding argument.
\end{proof}

\subsection{Transferring the coding construction}\label{sec:pinching-transfer}
The third step transfers the variational coding construction to the
original channel.

\begin{proposition}[Transfer to the original channel]
\label{alt:operational-bridge}
For every finite cqq output ensemble $\omega_{XBE}$ of a channel
$\mathcal W:A\to BE$, the pinched ensembles from
Proposition~\ref{alt:controlled-comparison} satisfy, for every $R\ge0$,
\begin{equation}
 E_{\mathrm{sc}}(\mathcal W,R)
 \le\lim_{n\to\infty}\sup_{\alpha>1}\frac{\alpha-1}{\alpha}
       \left[R-\frac1n\widetilde P_\alpha^{(1)}(\widehat\omega_n)\right].
 \label{alt:bridge-bound}
\end{equation}
\end{proposition}
The word-dependent center pinches can be removed when Bob's acceptance
tests have the same symmetry. The remaining common pinches have a
controlled fidelity cost.

\begin{lemma}[Invariant acceptance tests]\label{alt:invariant-tests}
Let $0\le M_z\le I_{B^n}$ commute with the representation of $G_z$ on
$B^n$. Then
\begin{equation}
 \Tr_{B^n}[(M_z\otimes I_{E^n})\widehat\omega_{B^nE^n,z}]
 =\Tr_{B^n}[(M_z\otimes I_{E^n})\omega'_{B^nE^n,z}].
 \label{alt:accepted-state-identity}
\end{equation}
\end{lemma}
\begin{proof}
Simultaneous $G_z$ invariance of $\omega'_z$ and invariance of $M_z$
imply that $\gamma_z:=\Tr_{B^n}[(M_z\otimes I)\omega'_z]$ commutes
with the representation on $E^n$. The central pinchings therefore
fix $M_z$ and $\gamma_z$. Self-adjointness of the Bob pinching under
the trace, followed by the Eve pinching, gives the stated identity.
\end{proof}

\begingroup
\mdfsetup{nobreak=true}
\begin{lemma}[Removing common output pinchings]\label{alt:accepted-lift}
Let $\mathcal E_B,\mathcal E_E$ be input-independent pinching maps,
with $v_E$ projections on Eve's system, and set
$\mathcal W'=(\mathcal E_B\otimes\mathcal E_E)\circ\mathcal W$.
Every $(1,M)$ code $\mathcal C'$ for $\mathcal W'$ gives a code
$\mathcal C$ for $\mathcal W$, with the same encoder, satisfying
\begin{equation}
 F(\mathcal C;\mathcal W)\ge v_E^{-1}F(\mathcal C';\mathcal W').
 \label{alt:common-pinching-fidelity}
\end{equation}
\end{lemma}
\endgroup
\begin{proof}
Pull Bob's decoder back through $\mathcal E_B$. If $\gamma_{E,m}$
is its correct-decoding block on $\mathcal W$, the corresponding
block on $\mathcal W'$ is $\mathcal E_E(\gamma_{E,m})$.
By concavity and averaging over the pinching unitaries, the optimizing
reference $\sigma_E$ for the latter code in
\eqref{eq:fidelity-blocks} can be chosen fixed by $\mathcal E_E$.
The unitary-average representation and Rotfel'd's inequality
\cite{Thompson} give
\begin{equation}
 \Tr\sqrt{\mathcal E_E(C)}\le\sqrt{v_E}\Tr\sqrt C,
 \qquad C\ge0.
\end{equation}
Apply this to $C=\sigma_E^{1/2}\gamma_{E,m}\sigma_E^{1/2}$,
average over $m$, and square.
\end{proof}

\begin{proof}[Proof of Proposition~\ref{alt:operational-bridge}]
Fix $n$ and write $B,E$ for $B^n,E^n$ and $z=x^n$ for a superletter.
Let $\mathcal W'_n=(\mathcal E_{B,n}\otimes\mathcal E_{E,n})
\circ\mathcal W^{\otimes n}$, with outputs $\omega'_z$.
By the symmetrization argument in Lemma~\ref{alt:universal-references},
the variational optimizations for $\widehat\omega_n$ may be restricted
to permutation-invariant auxiliary states $\theta_{ZBE}$.
Fix one such $\theta\ll\widehat\omega_n$, with distribution $q_z$, and put
\begin{align}
 I_\theta&:=I(Z:B)_\theta-I(Z:E)_\theta,\\
 d_\theta&:=\sum_zq_zD(\theta_{BE,z}\|\widehat\omega_{BE,z}).
\end{align}
We prove
\begin{equation}
 E_{\mathrm{sc}}(\mathcal W,R)
 \le\frac{d_\theta+[nR-I_\theta]_++\log v_{E,n}}n.
 \label{alt:symmetric-auxiliary-bound}
\end{equation}
For $I_\theta\le0$ this follows from $E_{\mathrm{sc}}\le R$.
Otherwise, choose $r>0$, $\ell>I(Z:E)_\theta$, and $\varepsilon>0$
such that $r+\ell<I(Z:B)_\theta-2\varepsilon$.

\emph{Symmetry-preserving tests and coding.}
For $k$ superletters, let $\Pi_k$ be the spectral typical projector
of $\theta_B^{\otimes k}$. For each $q$-typical word $z^k$, choose
the conditional spectral typical projector $\Pi_{z^k}$ of
$\theta_{B^k,z^k}$, with the typicality tolerances small enough that
\begin{equation}
 \Pi_{z^k}\le
 e^{k[H(B|Z)_\theta+\varepsilon]}\theta_{B^k,z^k},
 \qquad
 \Pi_k\theta_B^{\otimes k}\Pi_k
 \le e^{-k[H(B)_\theta-\varepsilon]}\Pi_k.
 \label{alt:typical-projector-bounds}
\end{equation}
Set $M_{z^k}=\Pi_k\Pi_{z^k}\Pi_k$ on these words and
$M_{z^k}=0$ on all other words. The standard typical-projector
estimates \cite[Section~II]{Devetak} give
\begin{align}
 \mathbb E_{q^{\otimes k}}\Tr M_{Z^k}\theta_{B^k,Z^k}
 &\longrightarrow1,\\
 \left\|\mathbb E_{q^{\otimes k}}M_{Z^k}\right\|_\infty
 &\le e^{-k(I(Z:B)_\theta-2\varepsilon)}.
 \label{alt:packing-norm}
\end{align}
The first follows by gentle measurement and the probability of the
typical set tending to one. For the second, sum the first inequality
in \eqref{alt:typical-projector-bounds} over typical words, extend the
positive sum to all words, and use
$\mathbb E_{q^{\otimes k}}\theta_{B^k,Z^k}=\theta_B^{\otimes k}$
and the second inequality there.
Since $\theta_B$ is permutation invariant and $\theta_{B,z}$
is $G_z$ invariant, both projectors commute with $\prod_jG_{z_j}$.
Thus Lemma~\ref{alt:invariant-tests} equates their accepted Eve
operators on $\omega'_{z^k}$ and $\widehat\omega_{z^k}$.

Draw an independent array $z(m,l)\sim q^{\otimes k}$, with
$M_k=\lceil e^{kr}\rceil$ messages and $L_k=\lceil e^{k\ell}\rceil$
randomization labels. For $\kappa\in\{\theta,\widehat\omega\}$ define
\begin{equation}
 R_\kappa:=\frac1{M_kL_k}\sum_{m,l}|m\rangle\langle m|\otimes
 \Tr_{B^k}[(M_{z(m,l)}\otimes I)\kappa_{B^kE^k,z(m,l)}].
\end{equation}
There are arrays satisfying simultaneously
\begin{align}
 \left\|R_\theta-\pi_{M_k}\otimes\theta_E^{\otimes k}\right\|_1
 &=o(1),\label{alt:accepted-approximation}\\
 \sum_{m,l}M_{z(m,l)}&\le cI_{B^k},
 \label{alt:array-packing}\\
 \frac1{M_kL_k}\sum_{m,l}
 D(\theta_{B^kE^k,z(m,l)}\|\widehat\omega_{B^kE^k,z(m,l)})
 &\le kd_\theta+o(k).
 \label{alt:array-relative-entropy}
\end{align}
Indeed, ordinary quantum covering \cite[Section~II]{Devetak} gives
the first in expectation before acceptance; acceptance changes it
by at most the mean rejection probability, which tends to zero.
The typical truncations used in the covering proof can be removed at
vanishing expected cost; the acceptance tests remain zero on atypical
words. Each test is a positive contraction, and
\eqref{alt:packing-norm} gives
\begin{equation}
 \left\|\mathbb E\sum_{m,l}M_{z(m,l)}\right\|_\infty
 \le e^{-k[I(Z:B)_\theta-2\varepsilon-r-\ell]+o(k)}.
 \label{alt:array-mean-bound}
\end{equation}
Matrix Chernoff \cite[Corollary~5.2]{Tropp} therefore gives, for every
fixed $c\ge1$ and all sufficiently large $k$,
\begin{equation}
 \Pr\!\left\{\sum_{m,l}M_{z(m,l)}\nleq cI_{B^k}\right\}
 \le\exp\!\left\{k\log\dim B
       -ck[I(Z:B)_\theta-2\varepsilon-r-\ell]+o(k)\right\}.
 \label{alt:array-chernoff}
\end{equation}
Choose $c[I(Z:B)_\theta-2\varepsilon-r-\ell]>\log\dim B$.
Thus \eqref{alt:array-packing} holds with probability tending to one.
Markov's inequality applied to the expected trace-distance error in
\eqref{alt:accepted-approximation}, together with this tail bound,
gives an event $G_k$ of probability tending to one on which both
properties hold. The nonnegative
cost in the third has expectation $kd_\theta$, so its conditional
expectation on $G_k$ is at most $kd_\theta/\Pr(G_k)=kd_\theta+o(k)$.
An array on $G_k$ therefore satisfies all three properties.

\emph{Fidelity transfer.}
Fix such an array and set $a=\Tr R_\theta$ and $\xi=R_\theta/a$.
Equation~\eqref{alt:accepted-approximation} gives $a\to1$ and
$T(\xi,\pi_{M_k}\otimes\theta_E^{\otimes k})=o(1)$.
Apply Lemma~\ref{lem:entropy-continuity} with conditioned system $M$
and $d_C=M_k$; its distance hypothesis holds for all sufficiently
large $k$. Since $\log M_k=kr+o(1)$, the resulting entropy error is
$o(k)$, and hence $H(M|E^k)_\xi=\log M_k-o(k)$.
Apply Lemma~\ref{lem:relative-entropy-fidelity} to the normalized
states with uniform array labels $(m,l)$ and outputs $\theta_{z(m,l)}$
and $\widehat\omega_{z(m,l)}$. The acceptance map applies the matching
test $M_{z(m,l)}$ and discards $B^k,l$, giving $R_\theta,R_{\widehat\omega}$.
Their input relative entropy is bounded by
\eqref{alt:array-relative-entropy}, so
\begin{equation}
 \log F(R_{\widehat\omega},\pi_{M_k}\otimes\xi_{E^k})
 \ge-kd_\theta-o(k).
\end{equation}
Alice chooses $l$ uniformly to encode $m$. Bob completes the sub-POVM
$\Lambda_m=c^{-1}\sum_lM_{z(m,l)}$ arbitrarily.
Each matching term in its correct-decoding state on
$(\mathcal W'_n)^{\otimes k}$ equals the corresponding term in
$c^{-1}R_{\widehat\omega}$ by Lemma~\ref{alt:invariant-tests};
the cross terms are positive. Monotonicity of fidelity in the
positive-semidefinite order and its homogeneity,
followed by Lemma~\ref{alt:accepted-lift}, give a physical $(nk,M_k)$ code with
\begin{equation}
 F(\mathcal C_k;\mathcal W)
 \ge v_{E,n}^{-k}c^{-1}
       F(R_{\widehat\omega},\pi_{M_k}\otimes\xi_{E^k}).
\end{equation}
Taking $k\to\infty$ proves
$E_{\mathrm{sc}}(\mathcal W,r/n)
\le(d_\theta+\log v_{E,n})/n$.
Let $\varepsilon\to0$, $\ell\searrow I(Z:E)_\theta$,
and $r\nearrow I_\theta$. Rate monotonicity and the $1$-Lipschitz
bound extend this to \eqref{alt:symmetric-auxiliary-bound}.

Finally, $d_\theta\le D(\theta\|\widehat\omega_n)$ by the relative-entropy
chain rule. Apply the minimax argument of
\eqref{alt:variational-exponent-minimax} on the compact convex set
of symmetric auxiliary states to obtain
\begin{equation}
 E_{\mathrm{sc}}(\mathcal W,R)
 \le\sup_{\alpha>1}\frac{\alpha-1}{\alpha}
       \left[R-\frac1n\widetilde P_\alpha^{(1)}(\widehat\omega_n)\right]
       +\frac{\log v_{E,n}}n.
 \label{alt:finite-block-bridge}
\end{equation}
The uniform estimate \eqref{alt:uniform-weighted-comparison} gives
convergence of the supremum on the right, and the last term vanishes
because $v_{E,n}$ is polynomial. Taking $n\to\infty$ proves the claim.
\end{proof}

\subsection{Regularization by blocking}\label{sec:regularized-achievability}
We apply the preceding construction to arbitrary input blocks and
optimize over their ensembles and lengths.

\begin{theorem}[Achievability bound]\label{shell-interpolation}
For every finite-dimensional channel $\mathcal W:A\to BE$ and every
$R\ge0$,
\begin{equation}
 E_{\mathrm{sc}}(\mathcal W,R)
 \le\sup_{\alpha>1}\frac{\alpha-1}{\alpha}
       [R-P_\alpha(\mathcal W)].
 \label{alt:regularized-achievability}
\end{equation}
\end{theorem}
\begin{proof}
To optimize over input ensembles and channel blocks before optimizing
the R\'enyi order, it is convenient to use the convex conjugate
\begin{equation}
 E_{\mathrm{sc}}^*(\mathcal W,u)
 :=\sup_{R\ge0}\{uR-E_{\mathrm{sc}}(\mathcal W,R)\},\qquad 0\le u\le1.
 \label{eq:conjugate-definition}
\end{equation}
Fix $0<u<1$ and set $\alpha=1/(1-u)$. For the pinched ensemble
constructed from any finite output ensemble $\omega$, evaluate
\eqref{alt:symmetric-auxiliary-bound} at $R=I_\theta/n$ when
$I_\theta>0$. When $I_\theta\le0$, use $E_{\mathrm{sc}}^*\ge0$.
Optimizing over the symmetric auxiliary states gives
\begin{equation}
 E_{\mathrm{sc}}^*(\mathcal W,u)
 \ge\frac{u\widetilde P_\alpha^{(1)}(\widehat\omega_n)
                   -\log v_{E,n}}n.
 \label{alt:bridge-conjugate}
\end{equation}
Here $d_\theta\le D(\theta\|\widehat\omega_n)$ supplies the full
relative-entropy penalty. Letting $n\to\infty$ in
\eqref{alt:uniform-weighted-comparison} yields
\begin{equation}
 E_{\mathrm{sc}}^*(\mathcal W,u)\ge uP_\alpha^{(1)}(\omega).
 \label{shell-main-bound}
\end{equation}
The block scaling \eqref{eq:exponent-block-scaling} implies
\begin{equation}
 E_{\mathrm{sc}}^*(\mathcal W^{\otimes k},u)
 =kE_{\mathrm{sc}}^*(\mathcal W,u).
 \label{eq:conjugate-block-scaling}
\end{equation}
Apply \eqref{shell-main-bound} to every finite output ensemble of
$\mathcal W^{\otimes k}$ and optimize over ensembles and $k$. By
Definition~\ref{def:regularized-private-information},
\begin{equation}
 E_{\mathrm{sc}}^*(\mathcal W,u)\ge uP_\alpha(\mathcal W).
 \label{alt:regularized-conjugate-bound}
\end{equation}

Lemma~\ref{lem:operational-convexity} implies that the exponent is
closed and convex, with supporting slopes in $[0,1]$. Its elementary
bounds make the conjugate finite on this interval. Being convex and
lower semicontinuous, the conjugate is continuous there, including at
the endpoints. Convex duality therefore gives
\begin{align}
 E_{\mathrm{sc}}(\mathcal W,R)
 &=\sup_{0<u<1}\{uR-E_{\mathrm{sc}}^*(\mathcal W,u)\},\\
 &\le\sup_{\alpha>1}\frac{\alpha-1}{\alpha}
       [R-P_\alpha(\mathcal W)].
\end{align}
This proves the achievability bound.
\end{proof}

\subsection{Secret-key generation has the same exponent}\label{sec:generation}

Unassisted secret-key generation allows Alice to prepare the key register
jointly with the channel input, rather than receive an externally supplied
uniform message \cite[Section~3.2]{WTB}. The target remains the uniform
ideal key in \eqref{eq:key-fidelity}.

\begin{corollary}\label{cor:generation}
For every pair of integers $n,M\ge1$, an $n$-use secret-key generation
code with an $M$-valued key and squared fidelity $f$ to the ideal key in
\eqref{eq:key-fidelity} admits an $(n,M)$ secret-key transmission code with squared
fidelity at least $f/4$.
Consequently, the exact exponent in
Theorem~\ref{thm:exact-quantum-exponent} also holds for unassisted
secret-key generation.
\end{corollary}
\begin{proof}
Dephasing both key registers cannot decrease fidelity to the ideal key.
A generation code is therefore specified by probabilities $p_m$,
conditional encoding states, and decoding effects $\Lambda_m$.
For a common reference $\sigma$, put
$f_m=\sqrt{F(\mathcal T_m^\dagger(\Lambda_m),\sigma)}$.
Here $\mathcal T_m^\dagger$ is defined by \eqref{eq:conditional-adjoint} for the
conditional output of message $m$. At this reference the squared fidelity is
$M^{-1}(\sum_m\sqrt{p_m}f_m)^2$.
Duplicate each supported encoding $k_m=\lceil Mp_m\rceil$ times and
divide its decoding effect by $k_m$. The number of labels is
$L=\sum_{p_m>0}k_m$, with $M\le L\le2M$.
For uniform messages, the resulting code has squared fidelity at least
\begin{align}
 \frac1{L^2}\left(\sum_m\sqrt{k_m}f_m\right)^2
 &\ge\frac{M^2}{L^2}\,
       \frac1M\left(\sum_m\sqrt{p_m}f_m\right)^2\\
 &\ge\frac1{4M}\left(\sum_m\sqrt{p_m}f_m\right)^2.
\end{align}
Retain the $M$ labels with largest individual root-fidelity contributions.
Their average cannot decrease. Complete the retained effects to a POVM
by assigning the unused effect to one outcome; positivity makes its
fidelity contribution nondecreasing. Choosing a fidelity-optimal $\sigma$
proves the stated conversion. Every secret-key transmission code is also a generation
code, so the two optimal fidelities differ by at most a factor four.
The same comparison holds after optimizing over $\log M\ge nR$;
dividing its logarithm by $n$ proves equality of the exponents.
\end{proof}

\section{Classical channels}\label{sec:classical}
A classical joint channel $W(b,e|a)$ is a quantum channel that measures
the input in a fixed basis and prepares the diagonal output state with
law $W(\cdot,\cdot|a)$. We show that its R\'enyi private information
has the variational form stated in Section~\ref{sec:main-results}.

Fix $\alpha>1$ and put $u=(\alpha-1)/\alpha$.
The norm representation \eqref{eq:radius-norm} gives a direct diagonal reduction. Every
positive output $A$ of $\mathcal T_x$ is diagonal, so sandwiched data processing
under Bob's dephasing $\Delta_B$ gives
\begin{equation}
 \|(\Delta_B\tau)^{-u/2}A(\Delta_B\tau)^{-u/2}\|_\alpha
 \le\|\tau^{-u/2}A\tau^{-u/2}\|_\alpha.
\end{equation}
This holds for every Eve reference and auxiliary family, so the
minimizing $\tau$ may be diagonal without exchanging optimizations.
In the dual quasi-norm representation of Lemma~\ref{lem:dual-norm},
every output of $\mathcal T_x^\dagger$
is diagonal; Eve's data-processing inequality then permits diagonal
$\sigma$. With $\tau$ diagonal, operator concavity
$\Delta_B(\eta^u)\le\Delta_B(\eta)^u$ permits diagonal $\eta$.
The minimizing $\nu$ in the quasi-norm variational formula is then
diagonal as well. Thus all the optimizations reduce to distributions.

The resulting scalar
optimization has the following variational form.
\begin{proposition}[Classical variational representation]\label{prop:classical-radius}
For every classical distribution $p_{XBE}=(p_{xbe})$ and $\alpha>1$,
\begin{equation}
 P_\alpha^{(1)}(p)
 =\max_{q\ll p}\left\{I(X:B)_q-I(X:E)_q-\frac{\alpha}{\alpha-1}D(q\|p)\right\}.
 \label{eq:classical-radius}
\end{equation}
\end{proposition}
\begin{proof}
Keep $u=(\alpha-1)/\alpha$. After the diagonal reduction, the trace sum
in Definition~\ref{def:radius} is a scalar moment. The minimax exchanges
in Lemma~\ref{lem:dual-norm} put its reference optimizations in the order
$\inf_\tau\sup_{\eta,\sigma}\inf_\nu$.
For fixed faithful references, Gibbs' variational formula gives
\begin{equation}
 \frac1u\log\sum_{x,b,e}p_{xbe}
 \left[\frac{\eta(b|x)\sigma(e)}{\tau(b)\nu(e|x)}\right]^u
 =\max_{q\ll p}\left\{
 \mathbb E_q\log\frac{\eta(B|X)\sigma(E)}{\tau(B)\nu(E|X)}
 -\frac1uD(q\|p)\right\}.
\end{equation}
At fixed $\tau,\sigma,\eta$, the expression is concave in $q$ and
convex in $\nu$, so Sion's theorem exchanges the maximum over $q$
and the infimum over $\nu$. The maxima over $q,\eta,\sigma$ then
commute. Optimizing the three auxiliary distributions gives
$\eta=q_{B|X}$, $\sigma=q_E$, and $\nu=q_{E|X}$, by faithful
approximation if necessary. Consequently the radius equals
$\inf_{\tau>0}\max_{q\ll p}G_\tau(q)$, where
\begin{equation}
 G_\tau(q)=I(X:B)_q-I(X:E)_q+D(q_B\|\tau)-\frac1uD(q\|p).
\end{equation}
For faithful $\tau$, the entropy part of $uG_\tau(q)$ is
\begin{equation}
 (1-u)H_q(XBE)+uH_q(E|XB)+uH_q(X|E),
\end{equation}
so $G_\tau$ is continuous and concave in $q$ on the compact simplex
supported on $\operatorname{supp}p$. It is continuous and convex in
$\tau$ on the convex set of faithful distributions.
Sion's minimax theorem in its one-compact-domain form \cite{Sion}
therefore interchanges the infimum and maximum.
Finally $\inf_{\tau>0}D(q_B\|\tau)=0$, including singular $q_B$ by
faithful approximation. This proves \eqref{eq:classical-radius}.
\end{proof}

For a classical channel $W$, Definition~\ref{def:channel-radius} reads
\begin{equation}
 P_\alpha^{(1)}(W)=\sup_{p_{XA}}P_\alpha^{(1)}(p_{XBE}),
 \qquad p_{XABE}=p_{XA}W.
 \label{eq:classical-channel-identification}
\end{equation}
We now verify the channel formula~\eqref{eq:classical-penalty-capacity}.
For the upper bound in~\eqref{eq:classical-penalty-capacity}, extend a candidate changed output law $q_{XBE}$ by the original posterior $p_{A|XBE}$. This preserves $D(q_{XBE}\|p_{XBE})$. The KL chain rule bounds that divergence below by the conditional channel-law defect in ~\eqref{eq:classical-penalty-capacity}. For the reverse bound, take any candidate $q_{XABE}$ in that formula and choose the physical prefix $p_{XA}=q_{XA}$. Discarding $A$ gives
\begin{align}
 D(q_{XBE}\|(q_{XA}W)_{XBE})
 &\le D(q_{XABE}\|q_{XA}W)\\
 &=D(q_{BE|XA}\|W_{BE|A}\mid q_{XA}).
\end{align}
The variational formula \eqref{eq:classical-radius} then gives the reverse inequality in~\eqref{eq:classical-penalty-capacity}. Both arguments apply to every blocklength, so the general exponent theorem applies directly to classical channels, including stochastic encodings.

\section{Conclusions}
We determined the exact strong converse exponent of squared joint-key
fidelity for finite-dimensional quantum wiretap channels in terms of a
regularized R\'enyi private information. A matching change-of-channel
and pinching argument and a tilted-state integral prove exponential fidelity
decay at every rate above private capacity. For classical channels,
the R\'enyi private information reduces to a variational expression
involving ordinary private information and a relative-entropy penalty.
The examples in Appendix~\ref{sec:criteria} show that joint-key fidelity and marginal
error criteria can have different exponents and must be treated separately.

An important open question is to identify classes of quantum wiretap
channels for which the R\'enyi private information, optimized over input
ensembles, is additive under tensor powers. For such channels, the exact
strong converse exponent would admit a single-letter characterization.
A second question is to determine the exact strong converse exponent of
$[p_{\rm succ}-\delta_T]_+$ and its relation to the squared-fidelity exponent.

\paragraph*{Statement on the use of artificial intelligence.}
The authors used OpenAI's ChatGPT and
Codex to assist with proof exploration, language editing, organization, literature searches, and
checks of mathematical derivations. The authors reviewed and verified all
resulting text, citations, and calculations and take full responsibility for
the content of the manuscript.

\paragraph*{Acknowledgments.}
HC acknowledges support from National Science and Technology Council (NSTC 115-2628-E-002-005, NSTC 114-2119-M-001-002, and NSTC 115-2124-M-002-014) and Ministry of Education (NTU-115V2016-1, NTU-CC115L893705, and NTU-115L900702).
CH received funding by the Deutsche Forschungsgemeinschaft (DFG, German
Research Foundation) -- 550206990, the Federal Ministry of Research,
Technology and Space (BMFTR), Germany, under the QC service center QUICS
(grant no.~13N17418), and by Germany's Excellence Strategy -- EXC 2123/2
QuantumFrontiers -- 390837967. MT is supported by the National Research Foundation, Singapore, via an NRF Investigatorship award (NRF-NRFI10-2024-0006).

\appendix
\section{On other error criteria}\label{sec:criteria}
The exact formula \eqref{eq:exact-quantum-exponent} concerns squared
joint-key fidelity for a specified joint channel. Independently
specified quantum marginals need not be compatible, whereas every
pair of classical channels admits a joint implementation. We use $p_{\rm succ},\delta_T$ from
\eqref{eq:gap-definition}. For an $(n,M)$ secret-key transmission code $\mathcal C$
with uniformly distributed messages, write the joint trace-distance
expression in Proposition~\ref{prop:oneshot-trace} as
\begin{equation}
 G_T(\mathcal C;\mathcal W)=1-\min_{\sigma_{E^n}}T(\rho_{M\widehat M E^n},
                    \kappa_{M\widehat M}\otimes\sigma_{E^n}).
\end{equation}
Fidelity data processing, the triangle inequality, and the Fuchs--van de Graaf bounds \cite{FuchsGraaf} give
\begin{align}
[p_{\rm succ}-\delta_T]_+
 &\le G_T(\mathcal C;\mathcal W)\le\sqrt{F(\mathcal C;\mathcal W)},\label{eq:criterion-comparison-start}\\
 1-\sqrt{1-F(\mathcal C;\mathcal W)}
 &\le G_T(\mathcal C;\mathcal W).
\label{eq:criterion-comparison}
\end{align}
Thus these criteria detect the same vanishing-error regime, but the bounds alone do not identify their strong converse exponents. Two examples show that the exponents can differ.

\begin{proposition}[Reliability--secrecy gap separation]\label{prop:gap-separation}
If Bob's marginal channel is a postprocessing of Eve's, the largest
reliability--secrecy gap $[p_{\rm succ}-\delta_T]_+$ at fixed $n,M$ equals $1/M$;
hence its exponent at rate $R$ is $R$. For the complementary outputs of the
binary erasure channel with erasure probability $1/2$,
\begin{equation}
 E_{\mathrm{sc}}(\mathcal W,\log2)\le\log(8/5)<\log2.
\end{equation}
\end{proposition}
\begin{proof}
Eve can simulate Bob's decoder. Its probability of guessing the actual message is $p_{\rm succ}$, whereas a key independent of Eve gives probability $1/M$. Trace-distance testing therefore gives $\delta_T\ge p_{\rm succ}-1/M$. A constant encoder and uniform guessing attain $[p_{\rm succ}-\delta_T]_+=1/M$.

For the erasure channel, both marginal outputs are binary erasures with probability $1/2$, so the first assertion applies. Send either basis bit, let Bob read it when present and guess uniformly otherwise. Eve's positive correct-decoding block for message $m\in\{0,1\}$ is
\(\tfrac12|e\rangle\langle e|+\tfrac14|m\rangle\langle m|\).
Against the common reference
\(\sigma_E=\tfrac45|e\rangle\langle e|+\tfrac1{10}(|0\rangle\langle0|+|1\rangle\langle1|)\),
the fidelity of this positive block is $5/8$. Product codes establish the displayed upper bound on the fidelity exponent.
\end{proof}

\begin{proposition}[Joint trace-distance separation]\label{prop:trace-separation}
For the $d$-symbol public-copy channel $W(b,e|a)=\mathbf1\{b=e=a\}$, with $d\ge2$,
\begin{align}
E_{\mathrm{sc}}(W,R)&=R,\label{eq:trace-separation-start}\\
 E_T(W,R)&:=\lim_{n\to\infty}-\frac1n\log
 \sup_{\mathcal C}G_T(\mathcal C;W)\\
 &=\max\{R/2,R-\log d\}.
\label{eq:trace-separation}
\end{align}
The supremum is over all $(n,M)$ secret-key transmission codes with $\log M\ge nR$,
as in \eqref{eq:optimal-fidelity}.
\end{proposition}
\begin{proof}
For a stochastic encoder $p_{y|m}$ and decoder $D(m|y)$ on $n$ uses, Eve observes Bob's public output $y$ among $D_0=d^n$ possibilities. Optimizing Eve's reference in the c-q fidelity formula gives
\begin{equation}
 F(\mathcal C;W)=\frac1{M^2}\sum_y
   \left(\sum_m\sqrt{p_{y|m}D(m|y)}\right)^2\le\frac1M,
\end{equation}
by Cauchy--Schwarz and $\sum_mD(m|y)=1$. A constant encoder and uniform guessing attain equality. For trace distance the overlap formula is
\begin{equation}
 G_T(\mathcal C;W)=\max_{\sigma}\frac1M\sum_{m,y}
              \min\{p_{y|m}D(m|y),\sigma_y\}.
\end{equation}
The inequality $\min(a,b)\le\sqrt{ab}$ followed by Cauchy--Schwarz yields $G_T(\mathcal C;W)\le M^{-1/2}$. Also $G_T(\mathcal C;W)\le p_{\rm succ}\le D_0/M$. Conversely, partition the messages into $L=\min\{D_0,\lfloor\sqrt M\rfloor\}$ groups of nearly equal size, send the group label, and guess uniformly within it. With $\sigma$ uniform on the used labels, every group contains at least $L$ messages and the overlap is $L/M$. The upper and lower estimates differ by at most a constant, proving \eqref{eq:trace-separation-start}--\eqref{eq:trace-separation}.
\end{proof}

\end{document}